\documentclass[10pt]{book}
\usepackage{array}

\usepackage{natbib}
\usepackage{bm}
\usepackage{setspace} 
\usepackage{amsmath,amssymb}
\usepackage{dsfont}
\usepackage{subfigure}
\usepackage{graphics}
\usepackage{epsfig}
\usepackage{multirow}
\usepackage{color}
\usepackage{times}
\usepackage{threeparttable}
\def\eqd{\,{\buildrel d \over =}\,}
\def\conind{\,{\buildrel d \over \rightarrow}\,}
\def\coninp{\,{\buildrel p \over \rightarrow}\,}
\usepackage{amsthm}
\newtheorem{thm}{Theorem}

\newtheorem{prop}{Proposition}

\begin{document}

%%%%%%%%%%%%%%%%%%%%%%%%%%%%%%%%%%%%%%%%%%%%%%%%%%%%%%%%%%%%%%%%%%%%%%%%%%%%%%%%%%%%%%%%%%%%%%%%%%%%%%%%%%%%%%%%%%%%%%%%%%%%
%%%%%%%%%%%%%%%%%%%%%%%%%%%%%%%%%%%%%%%%%%%%%%%%%%%%%%%%%%%%%%%%%%%%%%%%%%%%%%%%%%%%%%%%%%%%%%%%%%%%%%%%%%%%%%%%%%%%%%%%%%%%

\renewcommand{\baselinestretch}{1.2}

\markright{
\hbox{\footnotesize\rm Statistica Sinica (2012): Preprint}\hfill
}

\markboth{\hfill{\footnotesize\rm ARTIN ARMAGAN, DAVID B. DUNSON AND JAEYONG LEE} \hfill}
{\hfill {\footnotesize\rm GENERALIZED DOUBLE PARETO SHRINKAGE} \hfill}

\renewcommand{\thefootnote}{}
$\ $\par

%%%%%%%%%%%%%%%%%%%%%%%%%%%%%%%%%%%%%%%%%%%%%%%%%%%%%%%%%%%%%%%%%%%%%%%%%%%%%%%%%%%%%%%%%%%%%%%%%%%%%%%%%%%%%%%%%%%%%%%%%%%%

\fontsize{10.95}{14pt plus.8pt minus .6pt}\selectfont
\vspace{0.8pc}
\centerline{\large\bf GENERALIZED DOUBLE PARETO SHRINKAGE}
\vspace{.4cm}
\centerline{Artin Armagan, David B. Dunson and Jaeyong Lee}
\vspace{.4cm}
\centerline{\it SAS Institute Inc., Duke University and Seoul National University}
\vspace{.55cm}
\fontsize{9}{11.5pt plus.8pt minus .6pt}\selectfont

%%%%%%%%%%%%%%%%%%%%%%%%%%%%%%%%%%%%%%%%%%%%%%%%%%%%%%%%%%%%%%%%%%%%%%%%%%%%%%%%%%%%%%%%%%%%%%%%%%%%%%%%%%%%%%%%%%%%%%%%%%%%

\begin{quotation}
\noindent {\it Abstract:}
We propose a generalized double Pareto prior for Bayesian shrinkage estimation and inferences in linear models. The prior can be obtained via a scale mixture of Laplace or normal distributions, forming a bridge between the Laplace and Normal-Jeffreys' priors. While it has a spike at zero like the Laplace density, it also has a Student's $t$-like tail behavior. Bayesian computation is straightforward via a simple Gibbs sampling algorithm. We investigate the properties of the maximum a posteriori estimator, as sparse estimation plays an important role in many problems, reveal connections with some well-established regularization procedures, and show some asymptotic results. The performance of the prior is tested through simulations and an application.\par

\vspace{9pt}
\noindent {\it Key words and phrases:}
 Heavy tails, high-dimensional data, $\small{\mbox{LASSO}}$, maximum a posteriori estimation, relevance vector machine, robust prior, shrinkage estimation.
\par
\end{quotation}\par

%%%%%%%%%%%%%%%%%%%%%%%%%%%%%%%%%%%%%%%%%%%%%%%%%%%%%%%%%%%%%%%%%%%%%%%%%%%%%%%%%%%%%%%%%%%%%%%%%%%%%%%%%%%%%%%%%%%%%%%%%%%%

\fontsize{10.95}{14pt plus.8pt minus .6pt}\selectfont

\setcounter{chapter}{1}
\setcounter{equation}{0} %-1
\noindent {\bf 1. Introduction}

There has been a great deal of work in shrinkage estimation and simultaneous variable selection in the frequentist framework. The $\small{\mbox{LASSO}}$ of Tibshirani (1996) has drawn much attention to the area, particularly after the introduction of $\small{\mbox{LARS}}$ (Efron et al. (2004)) due to its superb computational performance. 
There is a rich literature analyzing the $\small{\mbox{LASSO}}$ and related approaches (Fu (1998), Knight and Fu (2000), Fan and Li (2001), Yuan and Lin (2005), Zhao and Yu (2006), Zou (2006), Zou and Li (2008)), with a number of articles considering asymptotic properties.  

Bayesian approaches to the same problem became popular with the works of Tipping (2001) and Figueiredo (2003). By expressing Student's $t$ priors for basis coefficients as scale mixtures of normals (West (1987)), and relying on type II maximum likelihood estimation (Berger (1985)), Tipping (2001) developed the relevance vector machine for sparse estimation in kernel regression.  In this setting, however, exact sparsity comes with the price of forfeiting propriety of the posterior by driving the scale parameter of the Student's $t$ distribution toward zero. In fact, driving both the scale parameter and the degrees of freedom to zero yields the so-called Normal-Jeffreys' prior, $\pi(\theta)\propto 1/|\theta|$. The name emerges due to the fact that the hierarchy follows as $\theta \sim \mbox{N}(0,\tau)$, $\pi(\tau)\propto 1/\tau$, where the latter is the Jeffreys' prior on the prior variance of $\theta$. Figueiredo (2003) proposed an expectation-maximization algorithm for maximum a posteriori estimation under Laplace and Normal-Jeffreys' priors, with estimates under the Laplace corresponding to the $\small{\mbox{LASSO}}$.  The Normal-Jeffreys' prior leads to substantially improved performance with finite samples due to the property of strongly shrinking small coefficients to zero while minimally shrinking large coefficients due to the heavy tails; however, it has no meaning from an inferential aspect as it leads to an improper posterior. 

A Bayesian $\small{\mbox{LASSO}}$ was proposed by Park and Casella (2008) and Hans (2009). However, these procedures inherit the problem of over-shrinking large coefficients 
due to the relatively light tails of the Laplace prior. Strawderman-Berger priors (Strawderman (1971), Berger (1980)) have some desirable properties yet lack a simple analytic form. Recently proposed priors have been designed to have high density near zero and heavy tails without the impropriety problem of Normal-Jeffreys.  The horseshoe prior of Carvalho, Polson, and Scott (2009, 2010) is induced through a carefully-specified mixture of normals, 
leading to such desirable properties as an infinite spike at zero and very heavy tails.  They studied sparse shrinkage estimation properties of the horseshoe in a normal means problem.  Griffin and Brown (2007, 2010) proposed an alternative class of hierarchical priors for shrinkage with some similarities to the prior we propose, but it lacks a simple analytic form that facilitates the study of some properties.     

There is a need for alternative shrinkage priors that lead to sparse point estimates if desired, do not over-shrink coefficients that are not close to zero, facilitate straightforward computation even in large $p$ cases, and result in a joint posterior distribution that does a good job of quantifying uncertainty.  We propose the generalized double Pareto prior which independently finds mention in Cevher (2009). It has a simple analytic form, yields a proper posterior, and possesses such appealing properties as a spike at zero, Student's $t$-like tails, and a simple characterization as a scale mixture of normals that leads to a straightforward Gibbs sampler for posterior inferences. We consider both fully Bayesian and frequentist penalized likelihood approaches based on this prior. We show that the induced penalty in the regularization framework yields a consistent thresholding rule having the continuity property in the orthogonal case, with a simple expectation-maximization algorithm described for sparse estimation in non-orthogonal cases. In another independent work motivated by applications to genome wide associations studies, Lee et al. (2011) consider a generalized $t$ prior (McDonald and Newey (1988)) that includes the generalized double Pareto as a special case. Similarities to previous work are limited and our contributions beyond them are (i) the formal introduction of a generalized Pareto density, thresholded and folded at zero, as a shrinkage prior in Bayesian analysis, (ii) the scale mixture representation of the generalized double Pareto in Proposition 1 which is central to our work, (iii) its connection to the Laplace and Normal-Jeffreys' priors as limiting cases in Proposition 2, (iv) the resulting fully conditional posteriors in a linear regression setting along with a simple Gibbs sampling procedure, (v) a detailed discussion on the hyper-parameters $\alpha$ and $\eta$ and their treatment, along with the incorporation of a griddy sampling scheme into the Gibbs sampler, (vi) a detailed analysis of the induced penalty by the generalized double Pareto prior and the properties of the resulting thresholding rule, (vii) an explicit analytic form for the maximum a posteriori estimator in orthogonal cases, (viii) an expectation-maximization procedure to obtain the maximum a posteriori estimate in non-orthogonal cases using the normal mixture representation, (ix) the one-step estimator (Zou and Li (2008)) resulting from the Laplace mixture representation, revealing the connection of the resulting procedure to the adaptive $\small{\mbox{LASSO}}$ of Zou (2006), and (x) the oracle properties of the resulting estimators. 
\par
\vspace{10pt}

\setcounter{chapter}{2}
\setcounter{equation}{0} %-1
\noindent {\bf 2. Generalized Double Pareto Prior}

The generalized double Pareto density is
\begin{equation}
f(\theta|\xi,\alpha)=\frac{1}{2\xi}\left(1+\frac{|\theta|}{\alpha\xi}\right)^{-(\alpha+1)},
\label{eq:gdP}
\end{equation}
where $\xi>0$ is a scale parameter and $\alpha>0$ is a shape parameter. In contrast to (\ref{eq:gdP}), the generalized Pareto density of Pickands (1975) is parametrized in terms of a location parameter $\mu\in \mathbb{R}$, a scale parameter $\xi>0$, and a shape parameter $\alpha\in \mathbb{R}$ as
\begin{equation}
f(\theta\, |\, \xi, \alpha, \mu) = \frac{1}{\xi} \bigg( 1 + \frac{ \theta - \mu}{\alpha \xi} \bigg)^{-(\alpha+1)},
\label{eq:gP}
\end{equation}
with $\theta \ge \mu$ for $\alpha > 0$ and $\mu \le \theta \le \mu - \xi \alpha$ for $\alpha < 0$. The mean and variance for the generalized Pareto distribution are $\mathbb{E}(\theta)=\mu+\xi/(1-1/\alpha)$ for $\alpha\notin [0,1]$ and $\mathbb{V}(\theta)=\xi^{2}(1-1/\alpha)^{-2}(1-2/\alpha)^{-1}$ for $\alpha\notin [0,2]$. If we let $\mu=0$, (\ref{eq:gP}) becomes an exponential density as $\alpha\rightarrow\infty$ with mean $\xi$ and variance $\xi^{2}$. 

To modify the generalized Pareto density to be a shrinkage prior, we let $\mu = 0$ and reflect the positive part about the origin, assuming $\alpha > 0$, for a density that is symmetric about zero. The mean and variance for the generalized double Pareto distribution are $\mathbb{E}(\theta)=0$ for $\alpha>1$ and $\mathbb{V}(\theta)=2\xi^{2}\alpha^{2}(\alpha-1)^{-1}(\alpha-2)^{-1}$ for $\alpha>2$. The dispersion is controlled by $\xi$ and $\alpha$, with $\alpha$ controlling the tail heaviness and $\alpha=1$ corresponding to Cauchy-like tails and no finite moments.

Figure \ref{fig1} compares the density in (\ref{eq:gdP}) to Cauchy and Laplace densities for the special case $\xi = \alpha = 1$, so that $f(\theta)=1/\{2(1+|\theta|)^2\}$. We refer to this form as the standard double Pareto. Near zero, the standard double Pareto resembles the Laplace density, suggesting similar sparse shrinkage properties of small coefficients in maximum a posteriori estimation. It also has Cauchy-like tails, which is appealing in avoiding over-shrinkage away from the origin. This is illustrated in Figure \ref{fig1}(a). Figure \ref{fig1}(b) illustrates
how the density in (\ref{eq:gdP}) changes for different values of $\xi$ and $\alpha$. 

Prior (\ref{eq:gdP}) can be represented as a scale mixture of normal distributions leading to computational simplifications.  As shorthand notation, let $\theta \sim \small{\mbox{GDP}}( \xi, \alpha )$ denote that $\theta$ has density (\ref{eq:gdP}).  

\begin{prop}
Let $\theta \sim \mbox{N}(0,\tau)$, $\tau \sim \mbox{Exp}(\lambda^{2}/2)$, and $\lambda \sim \mbox{Ga}(\alpha,\eta)$, where $\alpha>0$ and $\eta>0$. The resulting marginal density for $\theta$ is $\small{\mbox{GDP}}(\xi = \eta/\alpha, \alpha)$. 
\end{prop}

Proposition 1 reveals a relationship between the prior in (\ref{eq:gdP}) and the prior of Griffin and Brown (2007), with the difference being that Griffin and Brown (2007) place a mixing distribution on $\lambda^2$ leading to a marginal density on $\theta$ with no simple analytic form. 

In Proposition 2 we show that the prior in (\ref{eq:gdP}) forms a bridge between two limiting cases -- Laplace and Normal-Jeffreys' priors. 

\begin{prop}
Given the representation in Proposition 1, $\theta \sim \small{\mbox{GDP}}(\xi = \eta/\alpha, \alpha)$ implies 
\begin{enumerate}
\item $f(\theta)\propto 1/|\theta|$ for $\alpha=0$ and $\eta=0$,
\item $f(\theta|\lambda')=(\lambda'/2)\exp{(-\lambda'|\theta|)}$ for $\alpha\rightarrow\infty$, $\alpha/\eta = \lambda'$ and $0<\lambda'<\infty$.
\end{enumerate}
\end{prop}
\begin{proof}
For the first item, setting $\alpha=\eta=0$ implies placing a Jeffreys' prior on $\lambda$, $\pi(\lambda)\propto 1/\lambda$. Integration over $\lambda$ yields $\pi(\tau)\propto 1/\tau$, which implies the Normal-Jeffreys' prior on $\theta$. For the second item, notice that $\pi(\lambda)=\delta(\lambda-\lambda')$, where $\delta(.)$ denotes the Dirac delta function, since $\lim_{\alpha\rightarrow\infty}\lim_{\alpha/\eta\rightarrow \lambda'}{\mathbb{E}}(\lambda)=\lambda'$ and $\lim_{\alpha\rightarrow\infty}\lim_{\alpha/\eta\rightarrow \lambda'}\mathbb{V}(\lambda)=0$. Thus, $\int_{0}^{\infty}(\lambda/2)\exp{(-\lambda|\theta|)}\delta(d\lambda)=(\lambda'/2)\exp{(-\lambda'|\theta|)}$.
\end{proof}

\begin{figure}[!t]
\centering \subfigure[]{
\begin{minipage}{.49\linewidth}
 \centering\includegraphics[width=1\textwidth]{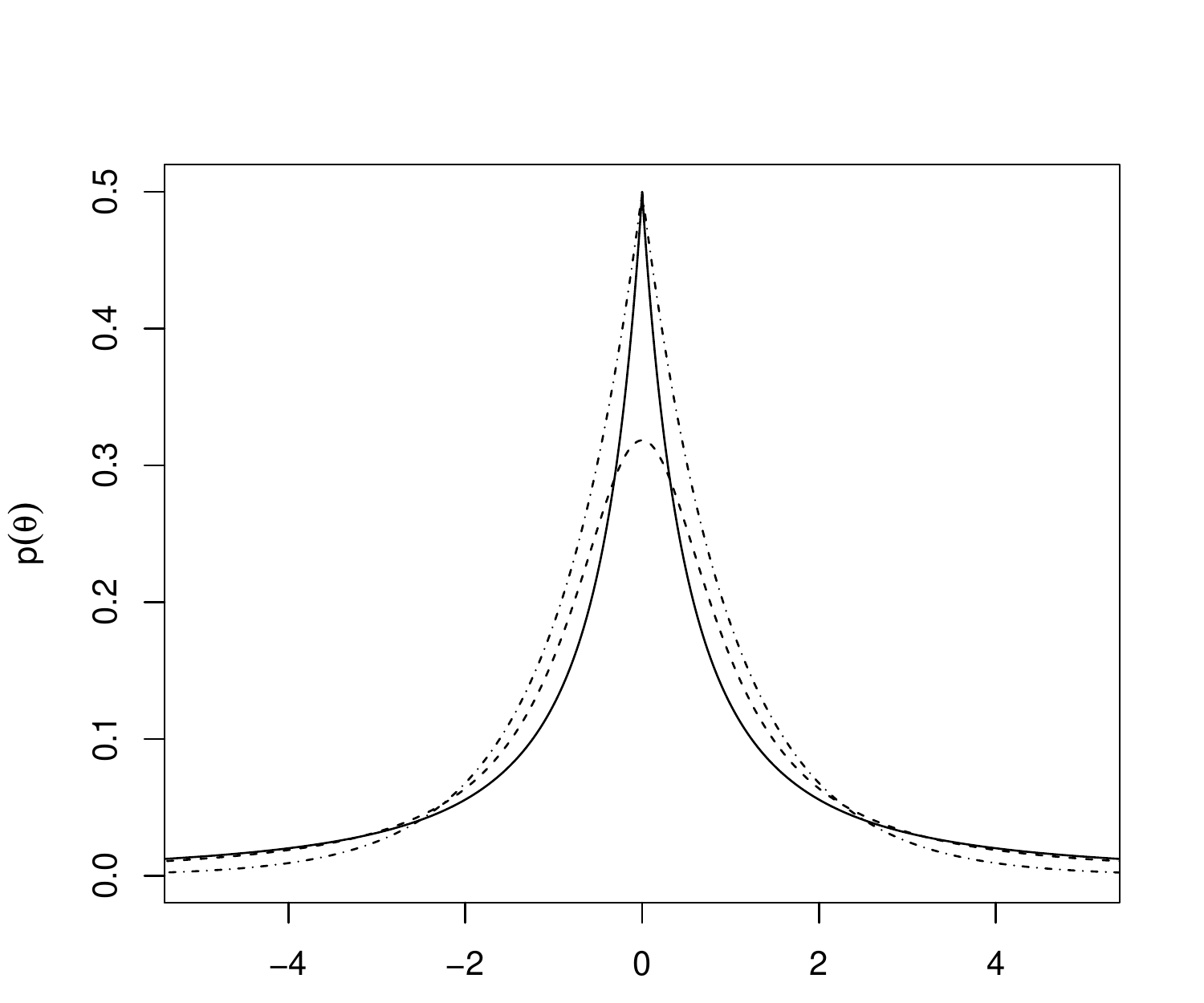} \\ \centering\includegraphics[width=1\textwidth]{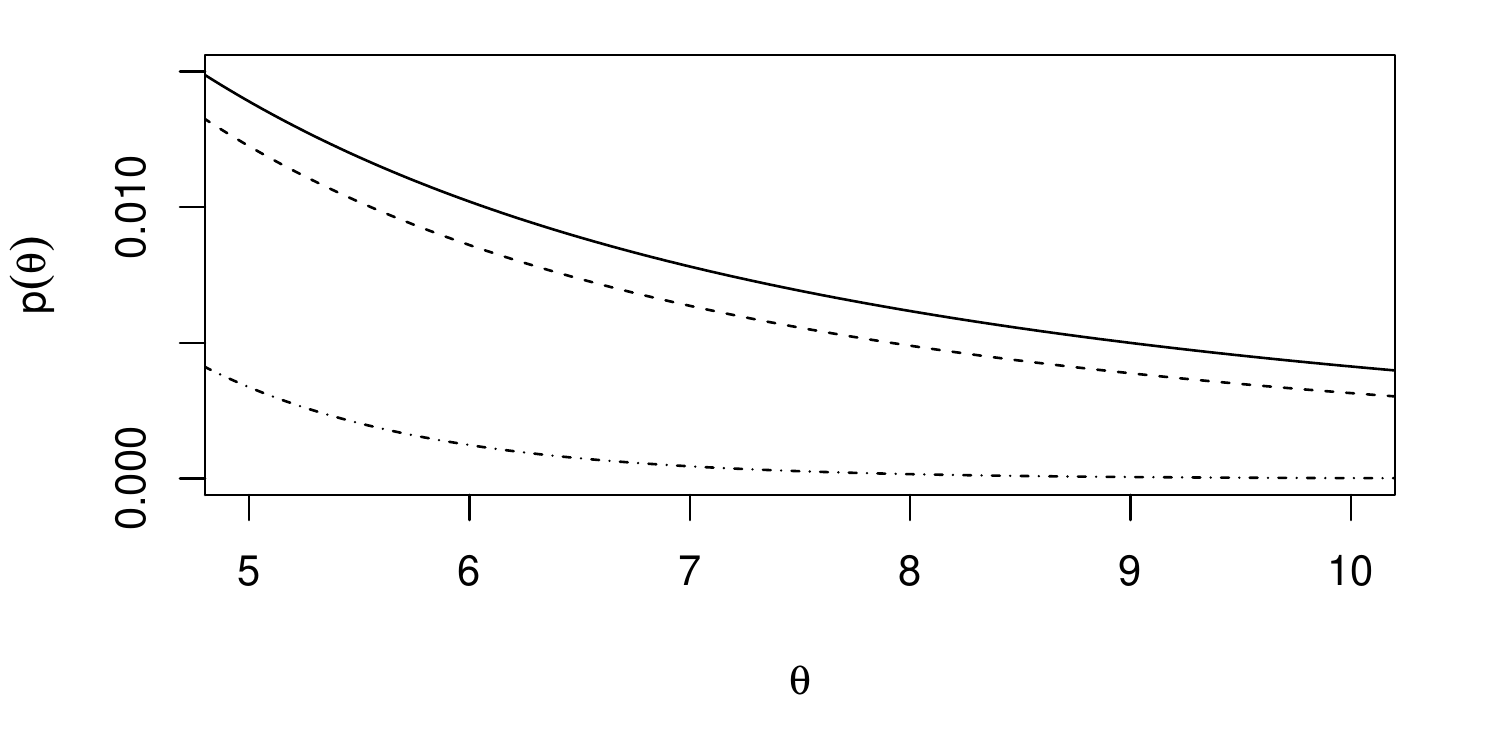}
\end{minipage}}
\centering \subfigure[]{
\begin{minipage}{.49\linewidth}
   \centering\includegraphics[width=1\textwidth]{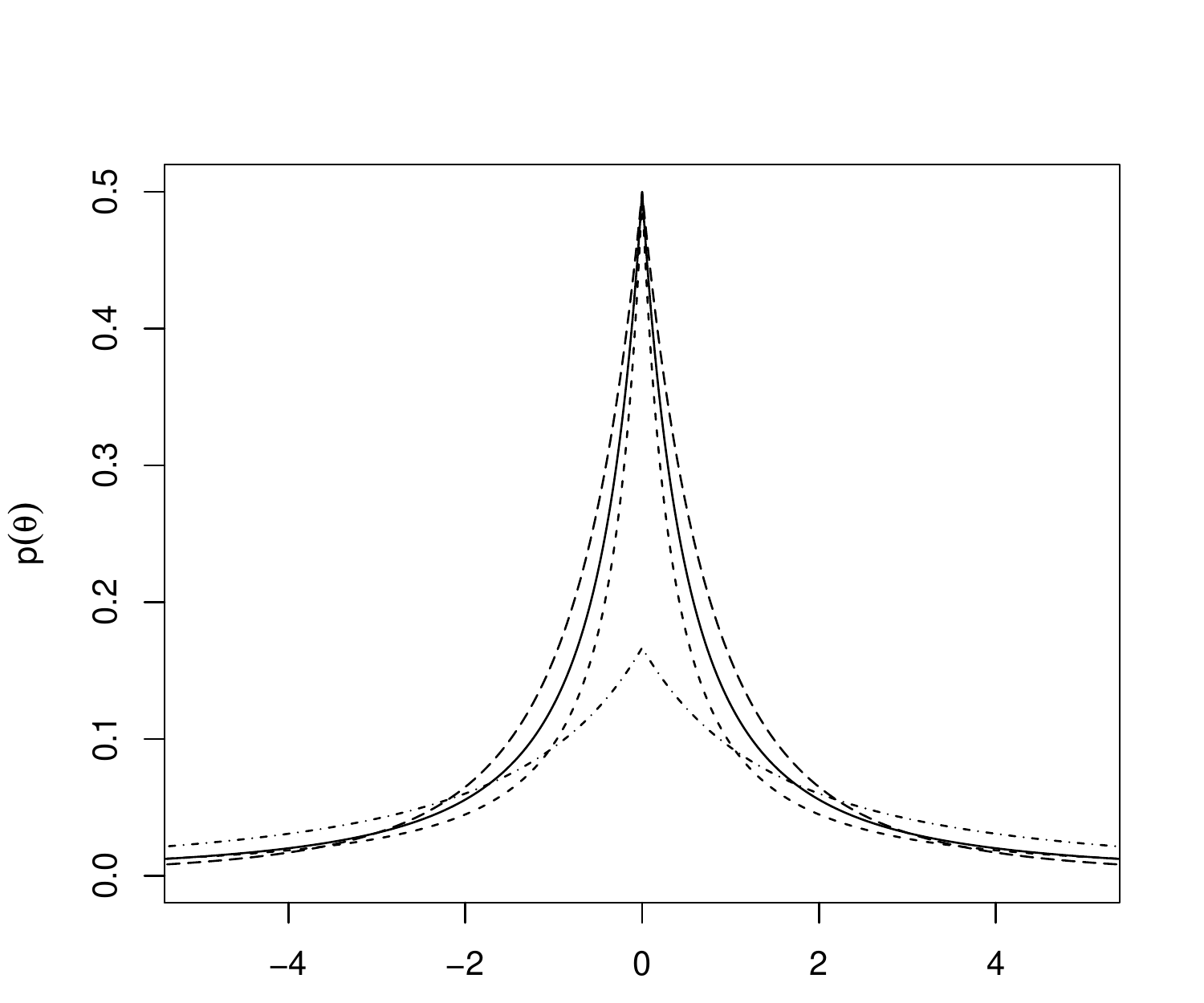} \\ \centering\includegraphics[width=1\textwidth]{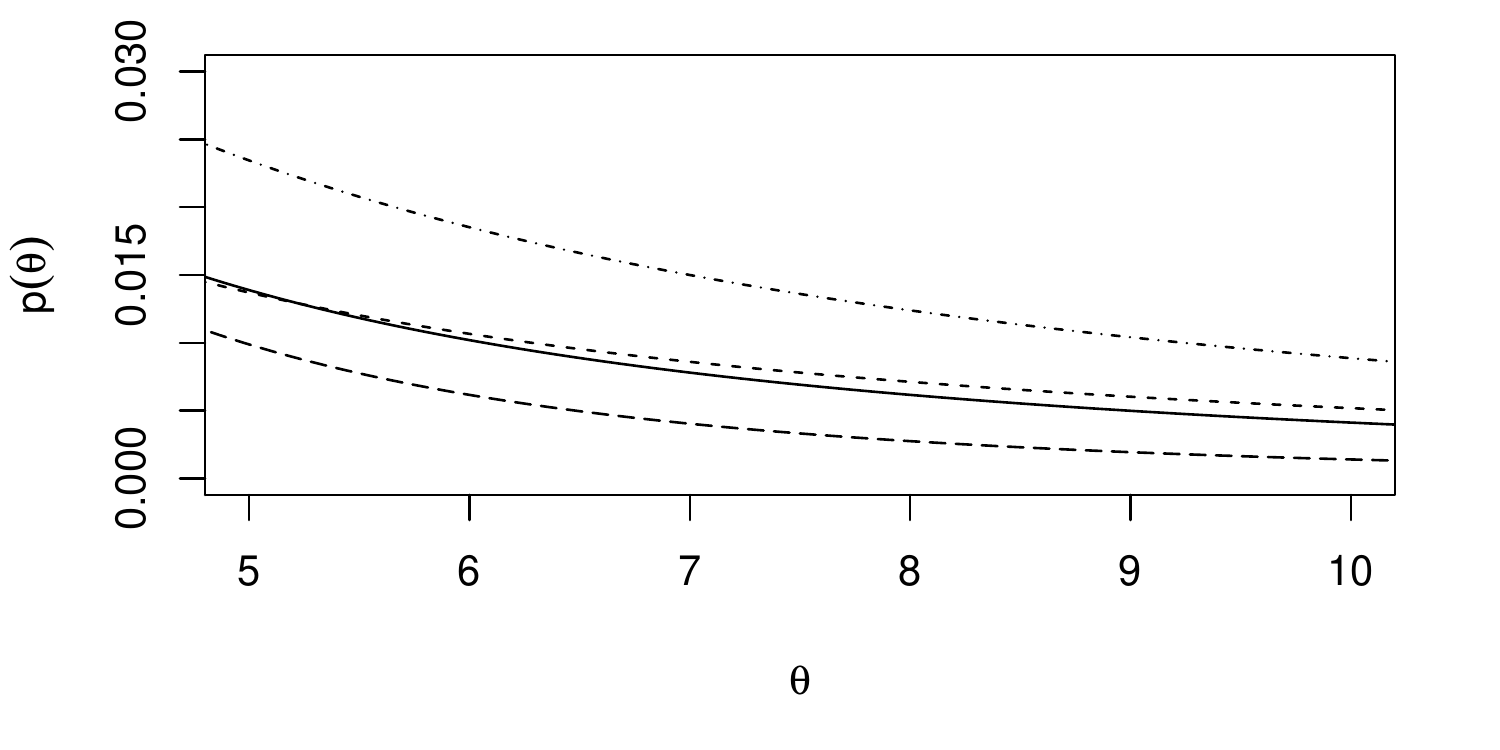}
\end{minipage}}
\caption{(a) Probability density functions for standard double Pareto (solid line), standard Cauchy (dashed line) and Laplace (dot-dash line) ($\lambda=1$) distributions. (b) Probability density functions for the generalized double Pareto with $(\xi,\alpha)$ values of $(1,1)$ (solid line), $(0.5,1)$ (dashed line), $(1,3)$ (long-dashed line), and $(3,1)$ (dot-dash line).\label{fig1}}
\end{figure}

As noted in Polson and Scott (2010), if $\pi(\tau)$ has exponential or lighter tails, observations are shrunk towards zero by some non-diminishing amount, regardless of size. This phenomenon is well-understood and commonly observed in estimation under the Laplace prior, where an exponential density mixes a normal density. The higher-level mixing (over $\lambda$) in Proposition 1 allows $\pi(\tau)$ to have heavier tails, remedying the unwanted bias. 

As $\alpha$ grows, the density becomes lighter tailed, more peaked and the variance becomes smaller, while as $\eta$ grows, the density becomes flatter and the variance increases. Hence if we increase $\alpha$, we may cause unwanted bias for large signals, though causing stronger shrinkage for noise-like signals; if we increase $\eta$ we may lose the ability to shrink noise-like signals, as the density is not as pronounced around zero; and finally, if we increase $\alpha$ and $\eta$ at the same rate, the variance remains constant but the tails become lighter, converging to a Laplace density in the limit. This leads to over-shrinking of coefficients that are away from zero. As a typical default specification for the hyper-parameters, one can take $\alpha = \eta = 1$. This choice leads to Cauchy-like tail behavior, which is well-known to have desirable Bayesian robustness properties. 

To motivate this default choice, we assess the behavior of the prior shrinkage factor $\kappa = 1/(1+\tau) \in (0,1)$, where $\theta \sim \mathrm{N}(0,\tau)$ is the parameter of interest (Carvalho et al. (2010)). As $\kappa\rightarrow 0$, the prior imposes no shrinkage, while as $\kappa\rightarrow 1$ it has a strong pull towards zero. The generalized double Pareto distribution implies a prior $\pi(\kappa)$ on $\kappa$ upon integration over $\lambda$ in Proposition 1. 
For the standard double Pareto, this is
\begin{equation}
\pi(\kappa) = \frac{1}{2(1-\kappa)^2}\left[\frac{\sqrt{\pi}\exp\left\{\frac{\kappa}{2(1-\kappa)}\right\}\mbox{Erfc}\left\{\sqrt{\frac{\kappa}{2(1-\kappa)}}\right\}}{\sqrt{2\kappa(1-\kappa)}}-1\right],\nonumber
\end{equation}
where $\mbox{Erfc}(.)$ denotes the complementary error function. 
In Figure 2.2, we compare $\pi(\kappa)$ under the standard double Pareto, Strawderman-Berger, horseshoe, and Cauchy priors, which may all be considered default choices. The priors behave similarly for $\kappa \approx 0$, implying similar tail behavior. The behavior of $\pi(\kappa)$ for $\kappa \approx 1$ governs the strength of shrinkage of small signals.  As $\kappa \to 1$, $\pi(\kappa)$ tends towards zero for the Cauchy, implying weak shrinkage, while $\pi(\kappa)$ is unbounded for the horseshoe, suggesting a strong pull towards zero for small signals.    The Strawderman-Berger and standard double Pareto priors are a compromise between these extremes, with $\pi(\kappa)$ bounded for $\kappa \to 1$ in both cases.  The standard double Pareto assigns higher density close to one than the Strawderman-Berger prior, and has the advantage of a simple analytic form over the Strawderman-Berger and horseshoe priors.
\begin{figure}[!t]
 \centering\includegraphics[width=1\textwidth]{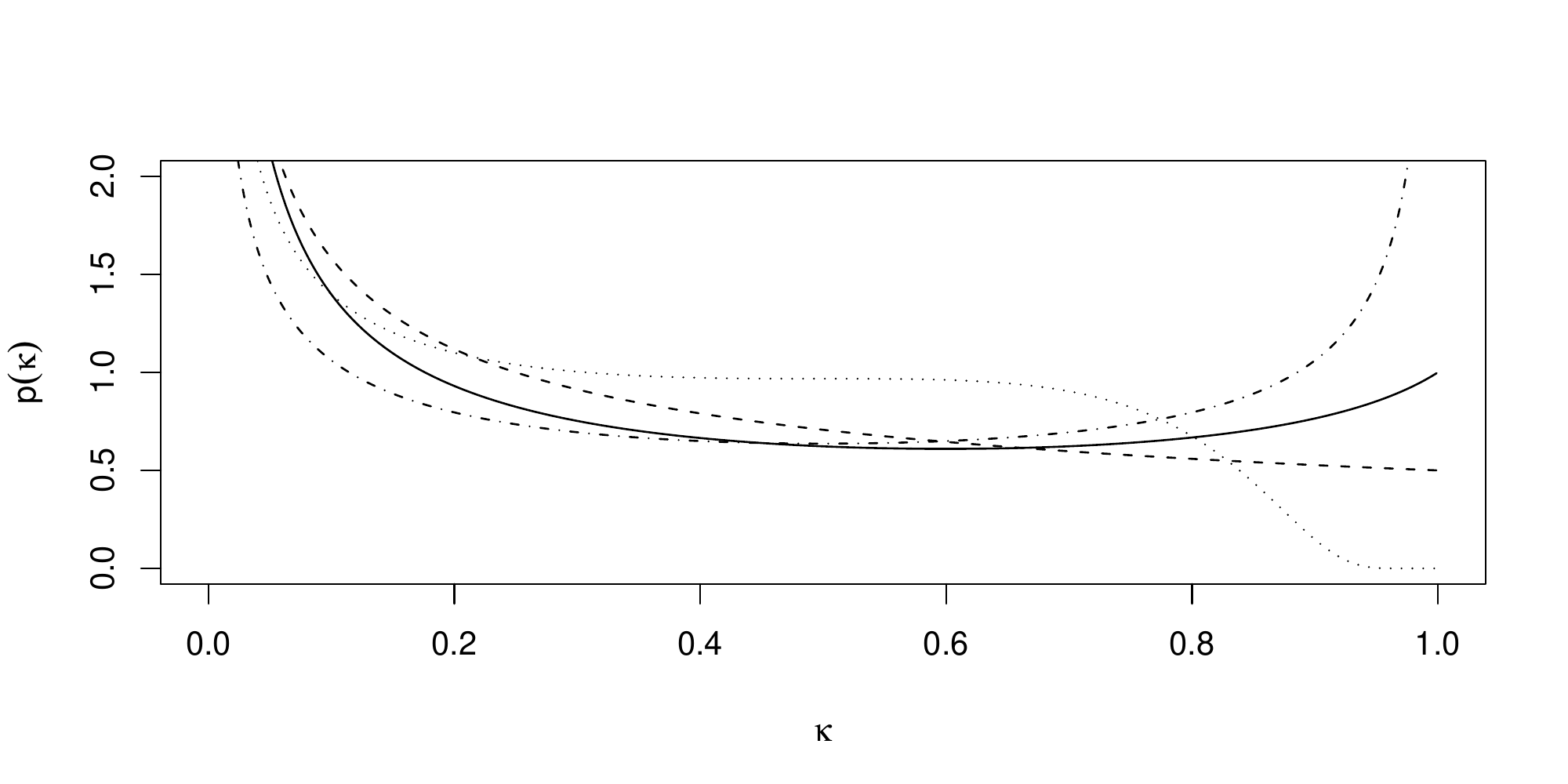}
\caption{Prior density of $\kappa$ implied by the standard double Pareto prior (solid line), Strawderman--Berger prior (dashed line), horseshoe prior (dot-dash line) and standard Cauchy prior (dotted line). \label{fig_shrink}}
\end{figure}

\begin{figure}[!t]
\centering \subfigure[]{
\begin{minipage}{.49\linewidth}
 \centering\includegraphics[width=1\textwidth]{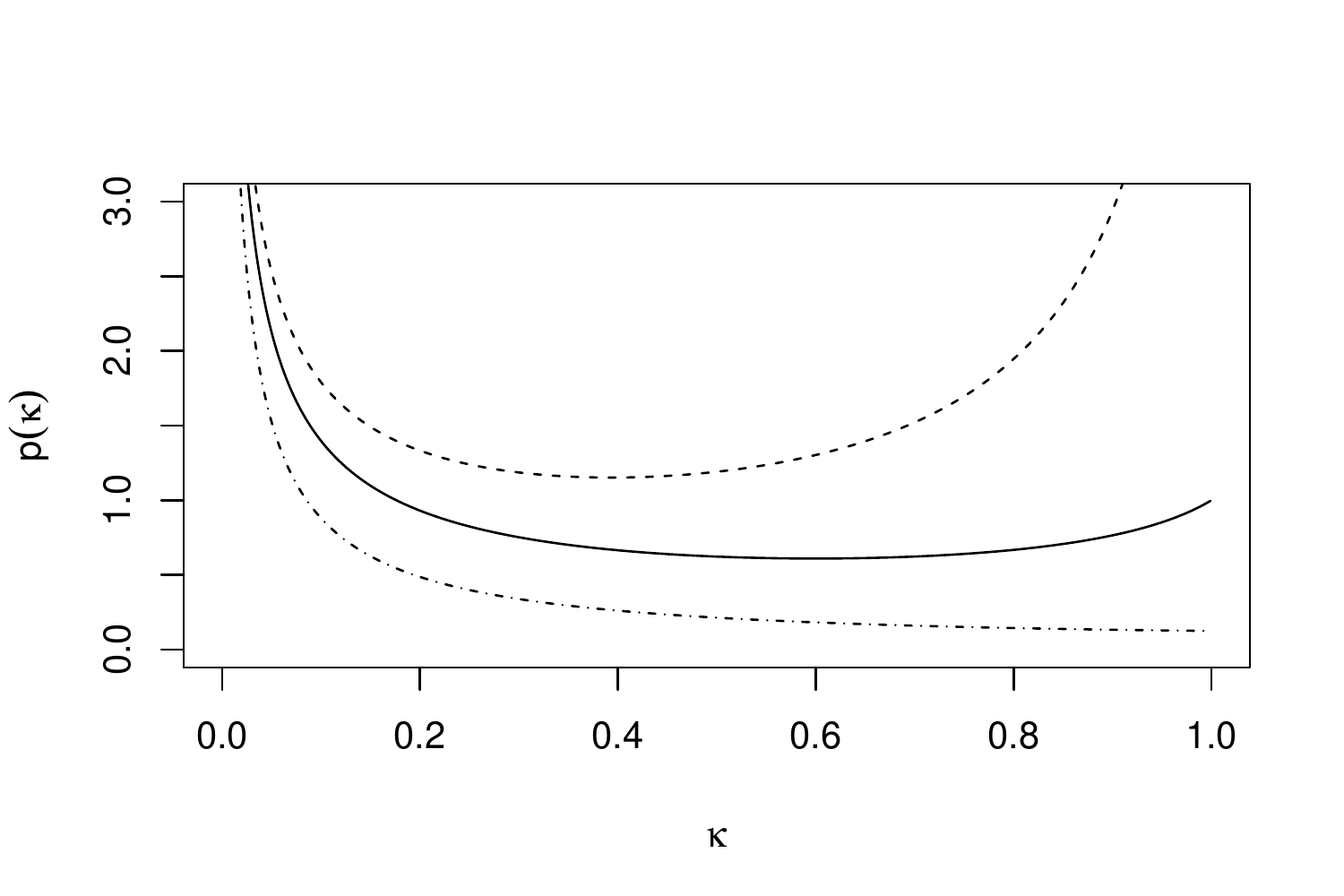}
\end{minipage}}
\centering \subfigure[]{
\begin{minipage}{.49\linewidth}
   \centering\includegraphics[width=1\textwidth]{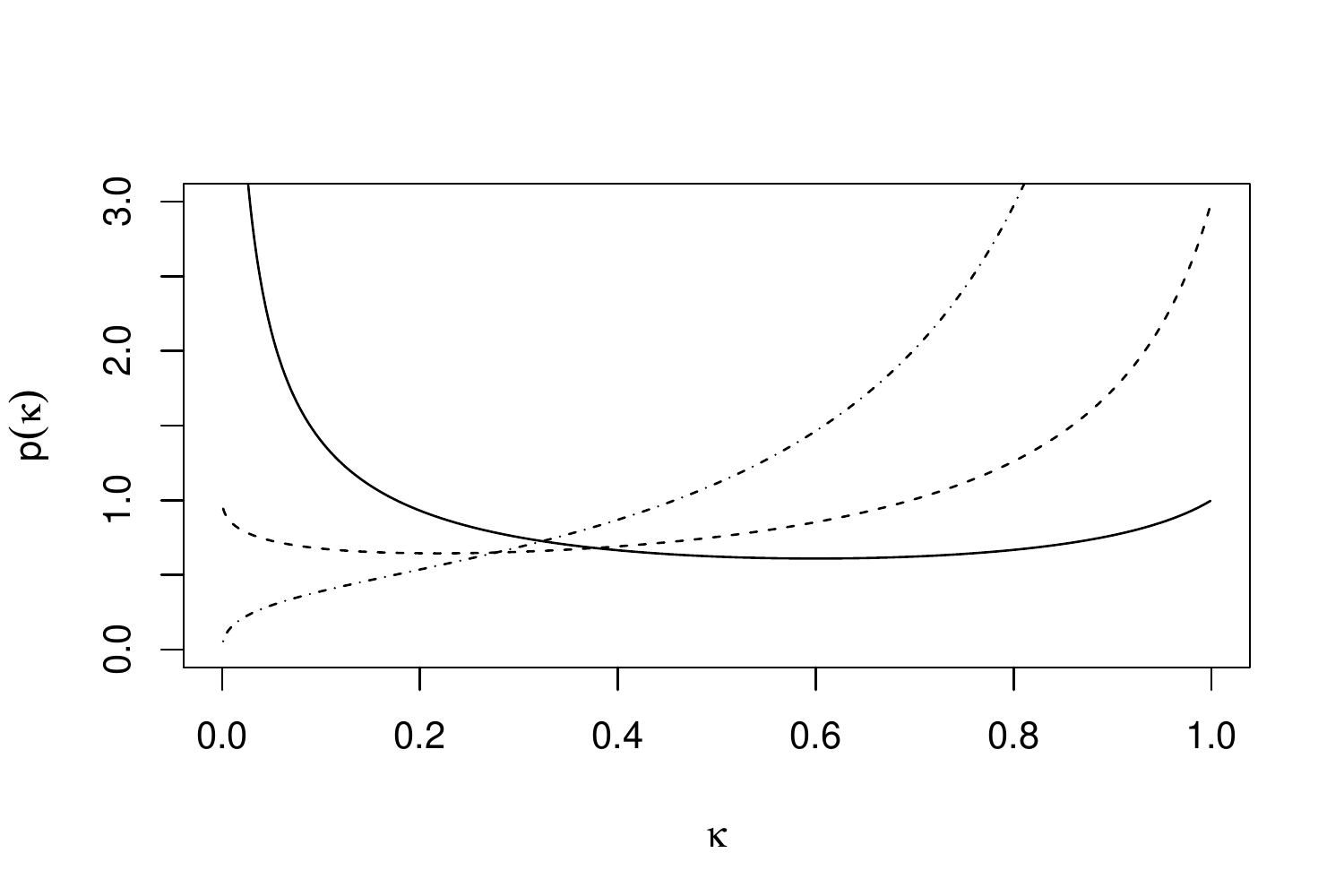}
\end{minipage}}
\caption{Prior density of $\kappa$ (a) when $\alpha=1$ and $\eta=0.5$ (dashed), $\eta=1$ (solid), $\eta=2$ (dot-dash) (b) when $\eta=1$ and $\alpha=1$ (solid), $\alpha=2$ (dashed), $\alpha=3$ (dot-dash).\label{fig_shrink2}}
\end{figure}

Of course it is best to adjust $\alpha$ and $\eta$ according to any available prior information pertaining to the sparsity structure of the estimated vector. For general $\alpha>0$ and $\eta>0$ values, the prior on $\kappa$ is
\begin{eqnarray}
\pi(\kappa|\alpha,\eta) &=& \frac{2^{\alpha/2-1}\eta^{\alpha}\kappa^{(\alpha-1)/2}(1-\kappa)^{-(\alpha+3)/2}}{\Gamma(\alpha)}\nonumber\\
& &\times\left\{\left(\frac{1}{\kappa}-1\right)^{1/2}\Gamma\left(\frac{\alpha}{2}+1\right){}_1\mbox{F}_{1}\left(\frac{\alpha}{2}+1,\frac{1}{2},\frac{\eta^2\kappa}{2(1-\kappa)}\right)\right.\nonumber \\
& &\left.-\sqrt{2}\eta\Gamma\left(\frac{\alpha+3}{2}\right){}_1\mbox{F}_{1}\left(\frac{\alpha+3}{2},\frac{3}{2},\frac{\eta^2\kappa}{2(1-\kappa)}\right)\right\},
\label{kappa}
\end{eqnarray}
where ${}_1\mbox{F}_1$ denotes the confluent hypergeometric function. Note that $\pi(\kappa|\alpha,\eta)$ takes a ``horseshoe'' shape when $\alpha=\eta=0$. Carvalho, Polson, and Scott (2010) show that $\pi(\kappa)\propto \kappa^{-1}(1-\kappa)^{-1}$ implies a Normal-Jeffreys' prior on $\theta$, which can also be observed by setting $\alpha=\eta=0$ in (\ref{kappa}) in conjunction with Proposition 1. Hence $\pi(\kappa|\alpha,\eta)$ is unbounded at $\kappa=1$ forcing $\pi(\theta|\alpha,\eta)$ to be unbounded at $0$ only if $\eta=0$. The effects of $\alpha$ and $\eta$ are now observed with better clarity from Figure \ref{fig_shrink2}. As $\eta$ increases, less and less density is assigned to the neighborhood of $\kappa\approx1$, repressing shrinkage. On the other hand, increasing $\alpha$ values place more and more density in the neighborhood of $\kappa\approx1$ promoting further shrinkage. This notion is later reinforced by Proposition 3, such that the prior induces a thresholding rule under maximum a posteriori estimation if $\eta<2\sqrt{\alpha+1}$. Hence, we need to carefully pick these hyper-parameters, in particular $\alpha$, as there is a trade-off between the magnitude of shrinkage and tail robustness.
\par
\vspace{10pt}

\setcounter{chapter}{3}
\setcounter{equation}{0} %-1
\noindent {\bf 3. Bayesian Inference in Linear Models}

Consider the linear regression model $\mathbf{y}=\mathbf{X}{\boldsymbol\beta}+{\boldsymbol\epsilon}$, where $\mathbf{y}$ is an $n$-dimensional vector of responses,
$\mathbf{X}$ is the $n\times p$ design matrix and $\boldsymbol\epsilon \sim \mbox{N}\left({0},\sigma^{2}\mathbf{I}_{n}\right)$.  Letting 
$\beta_j|\sigma \sim \small{\mbox{GDP}}(\xi = \sigma \eta/\alpha,  \alpha)$ independently for $j=1,\ldots, p$, 
\begin{equation}
\pi(\boldsymbol\beta|\sigma)=\prod_{j=1}^{p}\frac{1}{2\sigma\eta/\alpha}\left(1+\frac{1}{\alpha}\frac{|\beta_{j}|}{\sigma \eta/\alpha}\right)^{-(\alpha+1)}.\label{eq:gdP2}
\end{equation}
From Proposition 1, this prior is equivalent to $\beta_{j}|\sigma \sim \mathrm{N}(0,\sigma^2\tau_{j})$, with $\tau_j \sim \mbox{Exp}(\lambda_{j}^2/2)$ and 
$\lambda_j \sim \mbox{Ga}(\alpha,\eta)$.  We place the Jeffreys' prior on the error variance, $\pi(\sigma)\propto 1/\sigma$. 

Using the scale mixture of normals representation, we obtain a simple data augmentation Gibbs sampler having the conditional posteriors $(\boldsymbol\beta|\sigma^2,\mathbf{T},\mathbf{y})\sim \mathrm{N}\{(\mathbf{X}'\mathbf{X}+\mathbf{T}^{-1})^{-1}\mathbf{X}'\mathbf{y},\sigma^{2}\left(\mathbf{X}'\mathbf{X}+\mathbf{T}^{-1}\right)^{-1}\}$, $(\sigma^{2}|\boldsymbol\beta,\mathbf{T},\mathbf{y})\sim \mbox{IG}\{(n+p)/2,(\mathbf{y}-\mathbf{X}\boldsymbol\beta)'(\mathbf{y}-\mathbf{X}\boldsymbol\beta)/2+\boldsymbol\beta'\mathbf{T}^{-1}\boldsymbol\beta/2\}$, $(\lambda_{j}|\beta_{j},\sigma^{2})\sim \mbox{Ga}(\alpha+1,|\beta_{j}|/\sigma+\eta)$, $(\tau_{j}^{-1}|\beta_{j},\lambda_{j},\sigma^{2})\sim \mbox{Inv-Gauss}\{\mu=(\lambda_{j}^{2}\sigma^{2}/\beta_{j}^{2})^{1/2},\rho=\lambda^{2}\}$, where $\mathbf{T}=\mbox{diag}(\tau_{1},\ldots,\tau_{p})$ and $\mbox{Inv-Gauss}$ denotes the inverse Gaussian distribution with location and scale parameters $\mu$ and $\rho$. In our experience, this Gibbs sampler is efficient with fast rates of convergence and mixing. 

In the absence of any prior information on $\alpha$ and $\eta$, one may either set them to their default values or, as an alternative, choose hyper-priors to allow the data to inform about the values of $\alpha$ and $\eta$. We use $\pi(\alpha)=1/(1+\alpha)^2$ and $\pi(\eta)=1/(1+\eta)^2$ to correspond to generalized Pareto hyper-priors with location parameter $0$, scale parameter $1$ and shape parameter $1$. The median value of the resulting distribution for $\alpha$ and $\eta$ is $1$, centered at the default choices suggested earlier, while the mean and variance do not exist. 

For sampling purposes, let $a=1/(1+\alpha)$ and $e=1/(1+\eta)$. These transformations suggest a uniform prior on $a$ and $e$ in $(0,1)$ given the generalized Pareto priors on $\alpha$ and $\eta$. Consequently, the conditional posteriors for $a$ and $e$ are  
\begin{eqnarray} 
\pi(a|\boldsymbol\beta,\eta)&\propto& \left(\frac{1-a}{a}\right)^{p}\prod_{j=1}^{p}\left(1+\frac{|\beta_{j}|}{\sigma\eta}\right)^{-1/a}, \nonumber \\
\pi(e|\boldsymbol\beta,\alpha)&\propto& \left(\frac{e}{1-e}\right)^{p}\prod_{j=1}^{p}\left\{1+e\frac{|\beta_{j}|}{\sigma(1-e)}\right\}^{-(\alpha+1)}.\nonumber
\end{eqnarray}
We propose the embedded griddy Gibbs (Ritter and Tanner (1992)) sampling scheme:
\begin{enumerate}
\item[i.] Form a grid of $m$ points $a^{(1)},\ldots,a^{(m)}$ in the interval $(0,1)$.
\item[ii.] Calculate $w^{(k)}=\pi(a^{(k)}|\boldsymbol\beta,\eta)$.
\item[iii.] Normalize the weights, $w_{N}^{(k)}=w^{(k)}/\sum_{k=1}^{m}w^{(k)}$.
\item[iv.] Draw a sample from the set $\{a^{(1)},\ldots,a^{(m)}\}$ with probabilities $\{w_{N}^{(1)},\ldots,w_{N}^{(m)}\}$, and set $\alpha=1/a-1$ to be used at the current iteration of the Gibbs sampler.
\end{enumerate}
Repeat the same procedure for $e$ and obtain a random draw for $\eta$. We also experiment with fixing $\eta$ as $1$ while treating $\alpha$ as unknown. In this case, the prior variance of $\boldsymbol\beta|\sigma^{2}$ is determined by $\alpha$. 

In what follows we establish the ties between the Bayesian approach we have taken and some frequentist regularization approaches. The simple analytic structure of the generalized double Pareto prior facilitates analyses while its hierarchical formulation leads to straight-forward computation. 
\par
\vspace{10pt}

\setcounter{chapter}{4}
\setcounter{equation}{0} %-1
\noindent {\bf 4. Sparse Maximum a Posteriori Estimation}

The generalized double Pareto distribution can be used not only as a prior in a Bayesian analysis, but also to induce a sparsity-favoring penalty in regularized least squares:
\begin{equation}
\tilde{\boldsymbol\beta}=\mbox{arg}\min_{\boldsymbol\beta}\left\{\frac{1}{2\sigma^2}\|\mathbf{y}-\mathbf{X}\boldsymbol\beta\|^{2}+\sum_{j=1}^{p}p(|\beta_{j}|)\right\},
\label{pls}
\end{equation}
where $\mathbf{X}$ is initially assumed to have orthonormal columns and $p(.)$ denotes the penalty function implied by the prior on the regression coefficients. Following Fan and Li (2001), let $\hat{\boldsymbol\beta}=\mathbf{X}'\mathbf{y}$, and denote the minimization problem in (\ref{pls}) for a component of $\boldsymbol\beta$ as 
\begin{equation}
\tilde{\beta_{j}}=\mbox{arg}\min_{\beta_{j}}\left\{\frac{1}{2}\left(\hat{\beta}_{j}-\beta_{j}\right)^{2}+\sigma^{2}p(|\beta_{j}|)\right\},
\label{pls_one}
\end{equation}  
with the penalty function $p(|\beta_{j}|)=(\alpha+1)\log\left(\sigma\eta+|\beta_{j}|\right)$ that simply retains the term in $-\log \pi(\beta_j|\alpha,\eta)$ that depends on $\beta_j$.

From Fan and Li (2001), a good penalty function should result in an estimator that is (i) nearly unbiased when the true unknown parameter is large, (ii) a thresholding rule that automatically sets small estimated coefficients to zero to reduce model complexity, and (iii) continuous in data ($\hat{\beta}_j$) to avoid instability in model prediction.  In the following, we show that the penalty function induced by prior (\ref{eq:gdP2}) may achieve these properties.
\par
\vspace{10pt}

\noindent {\bf 4.1. Near-unbiasedness}

The first order derivative of (\ref{pls_one}) with respect to $\beta_{j}$ is $\mbox{sgn}(\beta_{j})\{|\beta_{j}|+\sigma^{2}p'(|\beta_{j}|)\}-\hat{\beta}_{j}=\mbox{sgn}(\beta_{j})\{|\beta_{j}|+\sigma^{2}(\alpha+1)/(\sigma\eta+|\beta_{j}|)\}-\hat{\beta}_{j}$, where $p'(|\beta_{j}|)=\partial p(|\beta_{j}|)/\partial |\beta_{j}|$ is the term causing bias in estimation. Although it is appealing to introduce bias in small coefficients to reduce the mean squared error and model complexity, it is also desirable to limit the shrinkage of large coefficients with $p'(|\beta_{j}|)\rightarrow 0$ as $|\beta_{j}|\rightarrow \infty$. In addition, it is desirable for $p'(|\beta_{j}|)$ to approach zero rapidly, implying shrinkage, and the associated introduction of bias rapidly decreases as coefficients get further away from zero. In fact, the rate of convergence of $p'(|\beta_{j}|)$ to zero is of the same order under the generalized double Pareto and Normal-Jeffreys' priors, with $\lim_{|\beta_{j}|\rightarrow \infty} \{(\alpha+1)/(\sigma\eta+|\beta_{j}|)\}/\{1/|\beta_{j}|\}=\alpha+1$. As $\alpha$ controls the tail heaviness in the generalized double Pareto prior, with lighter tails for larger values of $\alpha$, convergence of the ratio to $(\alpha+1)$ is intuitive. In the case of $\small{\mbox{LASSO}}$, the bias, $p'(|\beta_{j}|)$, remains constant regardless of $|\beta_{j}|$, which can also be observed in Figure \ref{fig3}(b).
\par
\vspace{10pt}

\noindent {\bf 4.2. Sparsity}

As noted in Fan and Li (2001), a sufficient condition for the resulting estimator to be a thresholding rule is that the minimum of the function $|\beta_{j}|+\sigma^{2}p'(|\beta_{j}|)$ is positive. 
\begin{prop}
Under the formulation in Proposition 1, prior (\ref{eq:gdP2}) implies a penalty yielding an estimator that is a thresholding rule if $\eta<2\sqrt{\alpha+1}$.
\end{prop}
This result is obtained by finding the minimum of $|\beta_{j}|+\sigma^{2}p'(|\beta_{j}|)$ and taking it greater than zero. The thresholding is a direct consequence of the fact that when $|\hat{\beta}_{j}|<\min_{\beta_{j}}\{|\beta_{j}|+\sigma^{2}(\alpha+1)/(\sigma\eta+|\beta_{j}|)\}$, which requires that $\min_{\beta_{j}}\{|\beta_{j}|+\sigma^{2}p'(|\beta_{j}|)\}>0$, the derivative of (\ref{pls_one}) is positive for all positive $\beta_{j}$ and negative for all negative $\beta_{j}$. In this case, the penalized least squares estimator is zero. When $|\hat{\beta}_{j}|>\min_{\beta_{j}}\{|\beta_{j}|+\sigma^{2}(\alpha+1)/(\sigma\eta+|\beta_{j}|)\}$, two roots may exist. The larger one (in absolute value) or zero is the penalized least squares estimator. To elaborate more on this, the root(s) may exist for $\mbox{sgn}(\beta_{j})\{|\beta_{j}|+\sigma^{2}p'(|\beta_{j}|)\}-\hat{\beta}_{j}=0$ only when $|\hat{\beta}_{j}|>\min_{\beta_{j}}\{|\beta_{j}|+\sigma^{2}p'(|\beta_{j}|)\}$. A helpful illustration is Figure 3 of Fan and Li (2001).
\par
\vspace{10pt}

\noindent {\bf 4.3. Continuity}

Continuity in data is important if an estimator is to avoid instabilities in prediction. As in Breiman (1996), ``a regularization procedure is unstable if a small change in data can make large changes in the regularized estimator''.  Discontinuities in the thresholding rule may result in inclusion or dismissal of a signal with minor changes in the data used (see Figure \ref{fig3}(b)). Hard-thresholding, the ``usual'' variable selection, is an unstable procedure, while ridge and $\small{\mbox{LASSO}}$ estimates are considered stable.

A necessary and sufficient condition for continuity is that the minimum of the function $|\beta_{j}|+\sigma^{2}p'(|\beta_{j}|)$ is at zero (Fan and Li (2001)). For our prior, the minimum of this function is obtained at $|\beta_{j}|=\sigma(\sqrt{\alpha+1}-\eta)$. Therefore $\eta=\sqrt{\alpha+1}$ yields an estimator with this property.
\begin{prop}
Under the formulation in Proposition 1, a subfamily of prior (\ref{eq:gdP}) with $\eta=\sqrt{\alpha+1}$ yields an estimator with the continuity property.
\end{prop}
In this particular case, the penalized likelihood estimator is set to zero if $|\hat{\beta}_{j}|\leq\sigma\sqrt{\alpha+1}$. When $|\hat{\beta}_{j}|>\sigma\sqrt{\alpha+1}$, 
\begin{equation}
\tilde{\beta}_{j}=\left\{\begin{array}{ll}
\frac{\hat{\beta}_{j}-\sigma\sqrt{\alpha+1}+\{\hat{\beta}_{j}^2+2\hat{\beta}_{j}\sigma\sqrt{\alpha+1}-3\sigma^2(\alpha+1)\}^{1/2}}{2} & \hat{\beta}_{j}>0, \\
\frac{\hat{\beta}_{j}+\sigma\sqrt{\alpha+1}-\{\hat{\beta}_{j}^2-2\hat{\beta}_{j}\sigma\sqrt{\alpha+1}-3\sigma^2(\alpha+1)\}^{1/2}}{2} & \hat{\beta}_{j}<0.
\end{array}\right.
\label{gdPMAP}
\end{equation} 
As can be observed in Figure \ref{fig3}(a), ensuring continuity by letting $\eta=\sqrt{\alpha+1}$ creates a trade-off between sparsity and tail-robustness. As the thresholding region becomes wider, the larger values are penalized further, yet not nearly at the level of $\small{\mbox{LASSO}}$.
\begin{figure}[!t]
\centering \subfigure[]{
\begin{minipage}{.49\linewidth}
 \centering\includegraphics[width=1\textwidth]{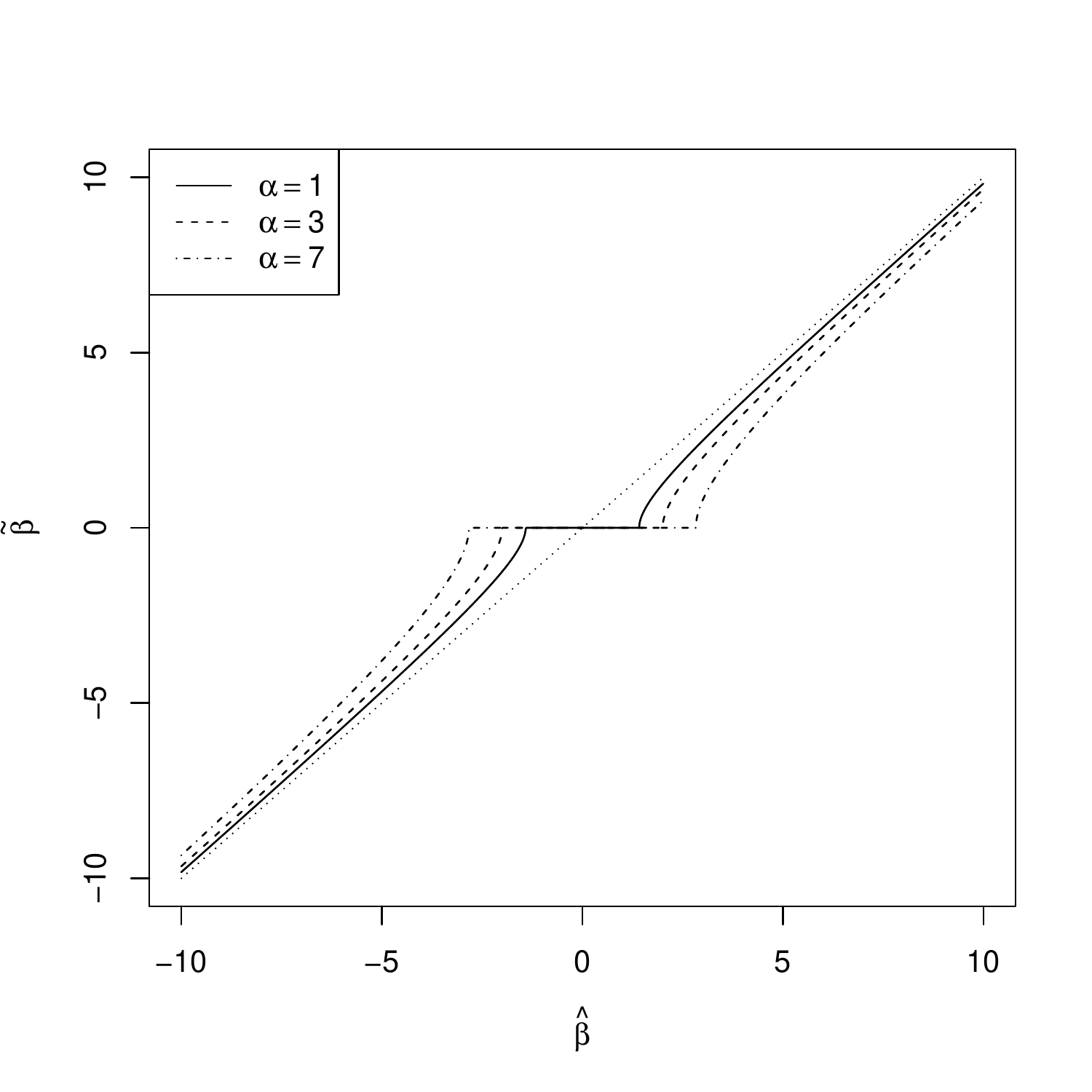}
\end{minipage}}
\centering \subfigure[]{
\begin{minipage}{.49\linewidth}
   \centering\includegraphics[width=1\textwidth]{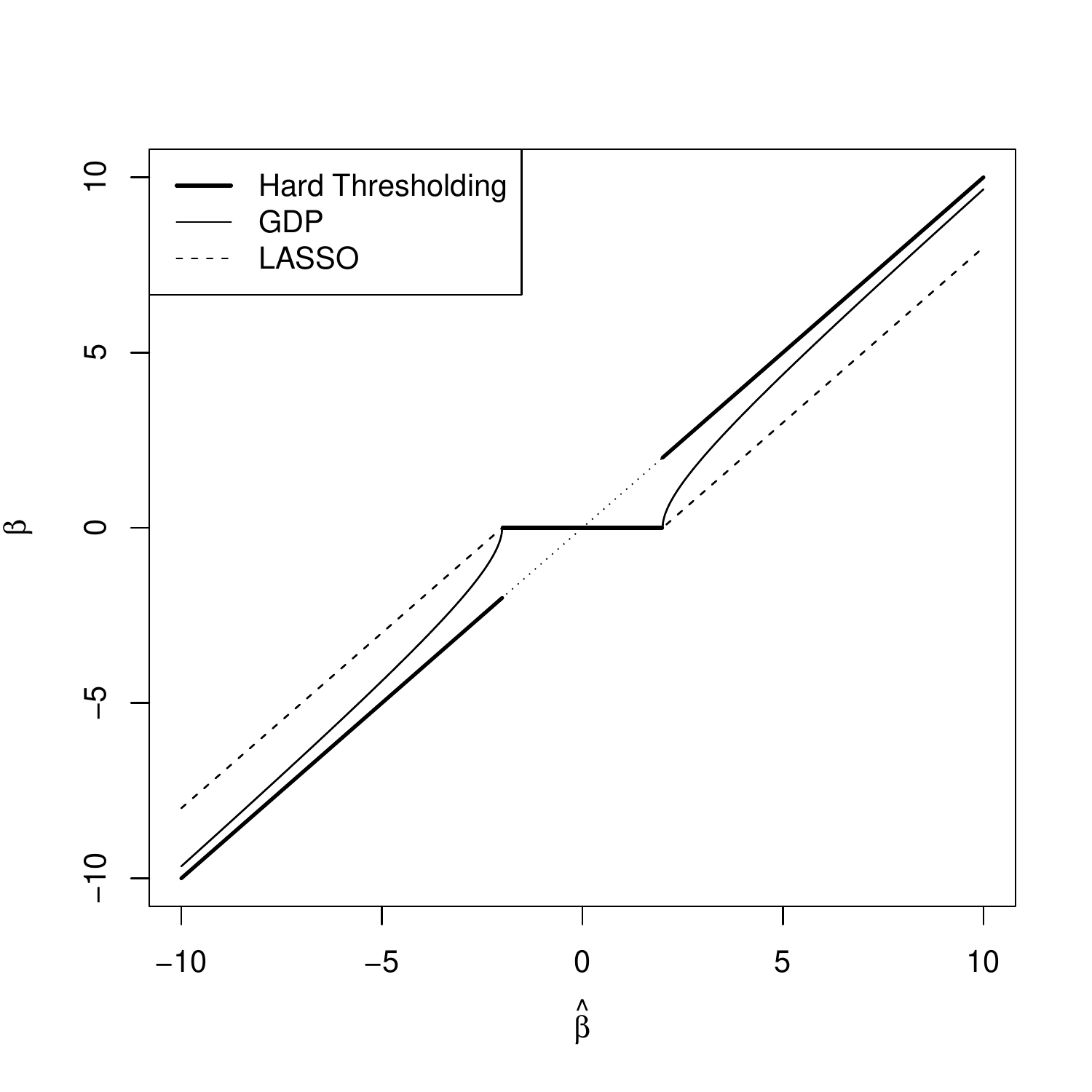}
\end{minipage}}
\caption{Thresholding functions for (a) generalized double Pareto prior with $\eta=\sqrt{\alpha+1}$, $\alpha=\{1,3,7\}$, (b) Hard thresholding, generalized double Pareto prior with $\eta=2$, $\alpha=3$ and $\small{\mbox{LASSO}}$ with $\sigma=1$. \label{fig3}} 
\end{figure}
\vspace{10pt}

\newpage

\noindent {\bf 4.4. Maximum a Posteriori Estimation via Expectation-Maximization}

We assume a normal likelihood to formulate the procedure for non-orthogonal linear regression. Estimation is carried out via the expectation-maximization ($\small{\mbox{EM}}$) algorithm.
\par
\vspace{10pt}

\noindent {\bf 4.4.1. Exploiting the Normal Mixture Representation}

We take the expectation of the log-posterior with respect to the conditional posterior distributions of $(\tau_{j}^{-1}|\beta_{j}^{(k)},\lambda_{j},\sigma^{2(k)})$ and $(\lambda_{j}|\beta_{j}^{(k)},\sigma^{2(k)})$ at the $k$th step, then maximize with respect to $\beta_{j}$ and $\sigma^{2}$ to get the values for the $(k+1)th$ step. 
%Removing the terms of the log-posterior that do not depend on $\boldsymbol\beta$ and $\sigma^{2}$, we are left with
%\begin{equation}
%-\left(\frac{n+p}{2}+1\right)\log\sigma^{2}-\frac{\left(\mathbf{y}-\mathbf{X}\boldsymbol\beta\right)'\left(\mathbf{y}-\mathbf{X}\boldsymbol\beta\right)-\sum_{j=1}^{p}\boldsymbol\beta_{j}^{2}/\tau_{j}}{2\sigma^{2}}.\nonumber
%\end{equation}
\begin{itemize}
\item \emph{E-step}:  
\begin{eqnarray}
& &-\left(\frac{n+p}{2}+1\right)\log\sigma^{2}-\frac{\left(\mathbf{y}-\mathbf{X}\boldsymbol\beta\right)'\left(\mathbf{y}-\mathbf{X}\boldsymbol\beta\right)}{2\sigma^{2}}\nonumber\\
& &\hspace{2in} -\frac{1}{2\sigma^{2}}\sum_{j=1}^{p}\beta_{j}^{2}\underbrace{\left\{\frac{(\alpha+1)\sigma^{2(k)}}{|\beta_{j}^{(k)}|(|\beta_{j}^{(k)}|+\sigma^{(k)}\eta)}\right\}}_{d_{j}^{(k)}}\nonumber
\end{eqnarray}
\item \emph{M-step}: Letting $\mathbf{D}^{(k)}=\mbox{diag}(d_{1}^{(k)},\dots,d_{p}^{(k)})$, we have
\begin{eqnarray}
\boldsymbol\beta^{(k+1)}&=&(\mathbf{X}'\mathbf{X}+\mathbf{D}^{(k)})^{-1}\mathbf{X}'\mathbf{y}, \nonumber\\ \sigma^{2(k+1)}&=&\frac{(\mathbf{y}-\mathbf{X}\boldsymbol\beta^{(k+1)})'(\mathbf{y}-\mathbf{X}\boldsymbol\beta^{(k+1)})+\boldsymbol\beta^{(k+1)'}\mathbf{D}^{(k)}\boldsymbol\beta^{(k+1)}}{n+p+2}.\nonumber
\end{eqnarray}
\end{itemize} 
We refer to this estimator as $\small{\mbox{GDP(MAP)}}$. 
\par
\vspace{10pt}

\noindent {\bf 4.4.2. Exploiting the Laplace Mixture Representation and the One-step Estimator}

In the proof of Proposition 1, the integration over $\boldsymbol\tau$ leads to a Laplace mixture representation of the prior. Since the mixing distribution of the Laplace is a known distribution the required expectation is obtained with ease, resulting in the maximization step,
\begin{eqnarray}
& &\boldsymbol\beta^{(k+1)}=\nonumber\\
& &\hspace{0.1in}\mbox{arg}\max_{\boldsymbol\beta}\left\{-\frac{1}{2\sigma^{2(k)}}\left(\mathbf{y}-\mathbf{X}\boldsymbol\beta\right)'\left(\mathbf{y}-\mathbf{X}\boldsymbol\beta\right)-\frac{1}{\sigma^{(k)}}\sum_{j=1}^{p}|\beta_{j}|\left(\frac{\alpha+1}{|\beta_{j}^{(k)}|/\sigma^{(k)}+\eta}\right)\right\}, \label{kstep} \nonumber \\
& & \\
& &\sigma^{2(k+1)}=\frac{b^2-2ac-\sqrt{b^4-4acb^2}}{2a^2},\nonumber
\end{eqnarray}
where $a = -(n+p+2)$, $b = (\alpha+1)\sum_j|\beta_j^{(k+1)}|/(|\beta_j^{(k)}|/\sigma^{(k)}+\eta)$, and $c = (\mathbf{y}-\mathbf{X}\boldsymbol\beta^{(k+1)})'(\mathbf{y}-\mathbf{X}\boldsymbol\beta^{(k+1)})$.
The component-specific multiplier on $|\beta_{j}|$ is obtained from the expectation of $\lambda_{j}$ with respect to its conditional posterior distribution, $\pi(\lambda_{j}|\beta_{j},\sigma^{2})$. Similar results to (\ref{kstep}) are in Candes, Wakin and Boyd (2008), Cevher (2009), and Garrigues (2009). 

An intuitive relationship to the adaptive $\small{\mbox{LASSO}}$ of Zou (2006) and the one-step sparse estimator of Zou and Li (2008) can be seen via the Laplace mixture representation. As a computationally fast alternative to estimating the exact mode via the above $\small{\mbox{EM}}$ algorithm, we can obtain a ``one-step estimator'' and exploit the $\small{\mbox{LARS}}$ algorithm as in Zou and Li (2008). The one-step estimator is
\begin{equation}
\boldsymbol\beta^{(1)}=\mbox{arg}\min_{\boldsymbol\beta}\left\{\left(\mathbf{y}-\mathbf{X}\boldsymbol\beta\right)'\left(\mathbf{y}-\mathbf{X}\boldsymbol\beta\right)+\alpha^{\dag}\sum_{j=1}^{p}\frac{|\beta_{j}|}{|\beta_{j}^{(0)}|+\eta^{\dag}}\right\},
\label{one_step}
\end{equation}
with $\alpha^{\dag}=2\sigma^{2(0)}(\alpha+1)$ and $\eta^{\dag}=\sigma^{(0)}\eta$. This estimator resembles the adaptive $\small{\mbox{LASSO}}$. The $\small{\mbox{LARS}}$ algorithm can be used to obtain $\boldsymbol\beta^{(1)}$ very quickly. We refer to this estimator as $\small{\mbox{GDP(OS)}}$.

\noindent {\bf Remark 1.} For $\eta^{\dag}=0$, the $\small{\mbox{GDP(OS)}}$ solution path for varying $\alpha^{\dag}$ is identical to the adaptive $\small{\mbox{LASSO}}$ solution path with $\gamma=1$ (see (4) in Zou (2006)) using identical $\boldsymbol\beta^{(0)}$.

\noindent {\bf Remark 2.} $\small{\mbox{GDP(OS)}}$ forms a bridge between the $\small{\mbox{LASSO}}$ and the adaptive $\small{\mbox{LASSO}}$: as $\eta^{\dag}\rightarrow \infty$ and $\alpha^{\dag}/\eta^{\dag}\rightarrow \lambda^{\dag}<\infty$, $\small{\mbox{GDP(OS)}}$ gives the $\small{\mbox{LASSO}}$ solution with penalty parameter $\lambda^{\dag}$.

We derive the $\small{\mbox{GDP(OS)}}$ estimator only to reveal a close connection with the adaptive $\small{\mbox{LASSO}}$ of Zou (2006) and do not use it in our experiments.
\par
\vspace{10pt}

\noindent {\bf 4.4.3. Normal vs. Laplace Representations in Computation}

As pointed out by an anonymous referee, it is appropriate to compare the convergence behavior of the $\small{\mbox{EM}}$ algorithms that exploit different mixture representations. We generated $n=\{200,400,600,800,1000\}$ observations from  $y_{i}=\mathbf{x}'_{i}\boldsymbol\beta^*+\epsilon_{i}$, where the $x_{ij}$ were independent standard normals for $p=\{20,40,60,80,100\}$, $\epsilon_{i}\sim \mbox{N}(0,\sigma^{2})$, and $\sigma=3$. We set the first $p/4$ components of $\boldsymbol\beta^*$ to be $1$ and the rest to $0$. For each $(n,p)$ combination we simulated $100$ data sets and ran the $\small{\mbox{EM}}$ algorithms obtained from normal and Laplace scale mixture representations. Figure \ref{fig4} illustrates the number of iterations taken by the two algorithms until $\|\boldsymbol\beta^{(k+1)}-\boldsymbol\beta^{(k)}\|_2<10^{-6}$. As expected, the convergence under the Laplace mixture representation was much faster with the intermediary mixing parameter $\tau_j$ integrated out rather than using the expectation step in the $\small{\mbox{EM}}$ algorithm.

\begin{figure}[!t]
 \centering\includegraphics[width=1\textwidth]{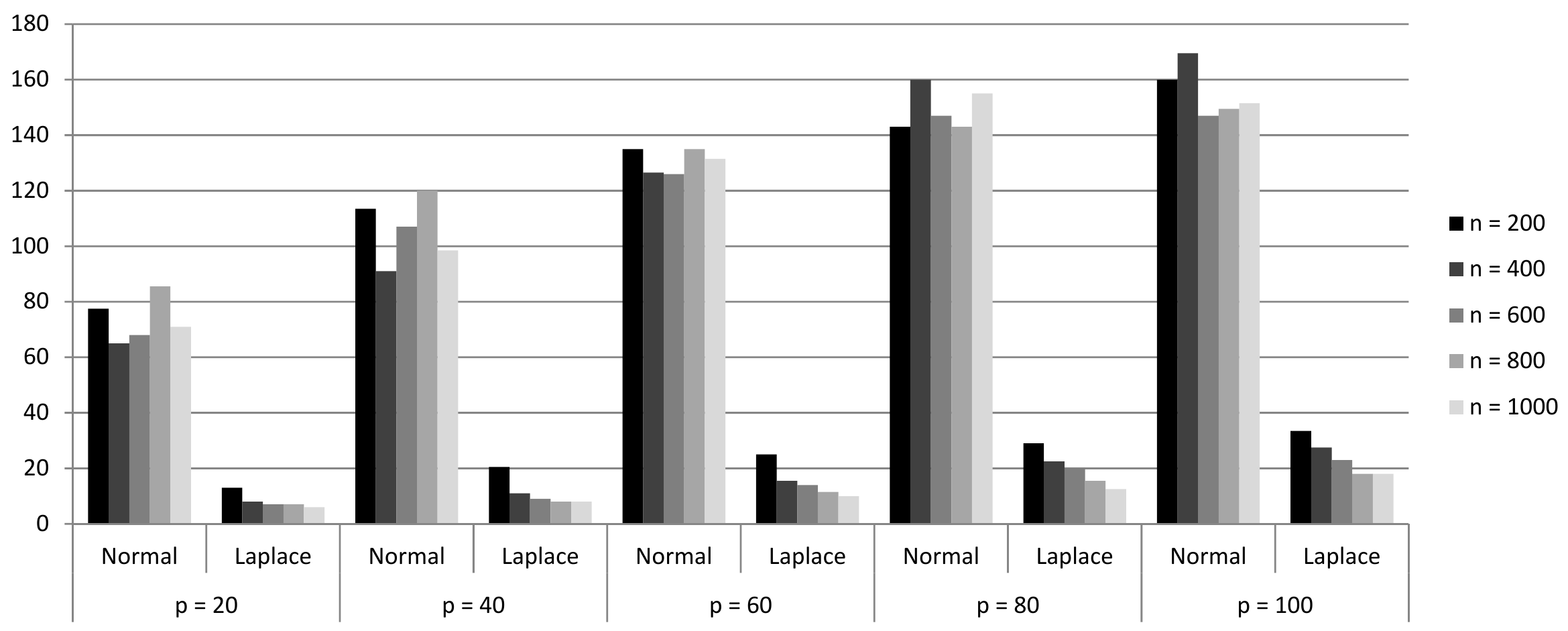}
\caption{Number of iterations until convergence of the $\small{\mbox{EM}}$ algorithms under normal and Laplace representations. \label{fig4}} 
\end{figure}
\par
\vspace{10pt}

\noindent {\bf 4.5. Oracle Properties}

Following Zou (2006) and Zou and Li (2008), we show that the $\small{\mbox{GDP(MAP)}}$ and $\small{\mbox{GDP(OS)}}$ estimators possess oracle properties. Relaxing the normality assumption on the error term leads to two conditions for Theorem 2 and Theorem 3.
\begin{enumerate}
\item[(A1)] $y_{i}=\mathbf{x}_{i}\boldsymbol\beta^{*}+\epsilon_{i}$ where $\epsilon_{1},\dots,\epsilon_{n}$ are independent and identically distributed with mean $0$ and variance $\sigma^{2}$.
\item[(A2)] $\frac{1}{n}\mathbf{X}'\mathbf{X}\rightarrow \mathbf{C}$, where $\mathbf{C}$ is a positive definite matrix.
\end{enumerate}
In what follows, $\mathcal{A}=\{j:\beta^{*}_{j}\neq 0, j=1,\dots,p\}$, $\boldsymbol\beta_{\mathcal{A}}$ retains the entries of $\boldsymbol\beta$ indexed by $\mathcal{A}$, and $\mathbf{C}_{\mathcal{A}}$ retains the rows and columns of $\mathbf{C}$ indexed by $\mathcal{A}$.

\begin{thm}
Let
\begin{equation}
\boldsymbol\beta^{(\infty)}_n=\textup{arg}\min_{\boldsymbol\beta}\left\{\left(\mathbf{y}-\mathbf{X}\boldsymbol\beta\right)'\left(\mathbf{y}-\mathbf{X}\boldsymbol\beta\right)+\alpha_n'\sum_{j=1}^{p}\log\left(|\beta_{j}|+\eta'_{n}\right)\right\}\nonumber
\end{equation}
denote the $\small{\mbox{GDP(MAP)}}$ estimator, where $\alpha'_{n}=2\sigma^2(\alpha_n+1)$ and $\eta'_{n}=\sigma\eta_n$. Let  $\mathcal{A}_{n}=\{j:\beta_{nj}^{(\infty)}\neq 0, j=1,\dots,p\}$. Suppose that $\alpha'_n\rightarrow\infty$, $\alpha'_{n}/\sqrt{n}\rightarrow0$ and, $\eta'_n\sqrt{n}\rightarrow c<\infty$. Then $\boldsymbol\beta^{(\infty)}_n$ is
\begin{enumerate}
\item consistent in variable selection in that $\lim_{n\rightarrow \infty}\mathbb{P}(\mathcal{A}_{n}=\mathcal{A})=1$;
\item asymptotically normal with $\sqrt{n}(\boldsymbol\beta^{(\infty)}_{n\mathcal{A}}-\boldsymbol\beta^{*}_{\mathcal{A}})\conind \mbox{N}(\mathbf{0},\sigma^{2}\mathbf{C}_{\mathcal{A}}^{-1})$.
\end{enumerate}
\end{thm}

\noindent{\bf Remark 3.} More generally, the above results hold if $\alpha'_n/(\sqrt{n}\eta'_n)\rightarrow\infty$ and $\alpha'_{n}/\sqrt{n}\rightarrow0$.

\begin{thm}
Let $\boldsymbol\beta^{(1)}_n$ denote the $\small{\mbox{GDP(OS)}}$ estimator in (\ref{one_step}) and $\mathcal{A}_{n}=\{j:\beta_{nj}^{(1)}\neq 0, j=1,\dots,p\}$. Suppose that $\alpha^{\dag}_{n}\rightarrow\infty$, $\alpha^{\dag}_{n}/\sqrt{n}\rightarrow0$, and $\eta^{\dag}_{n}\sqrt{n}\rightarrow c<\infty$. Then $\boldsymbol\beta^{(1)}_n$ is
\begin{enumerate}
\item consistent in variable selection in that $\lim_{n\rightarrow \infty}\mathbb{P}(\mathcal{A}_{n}=\mathcal{A})=1$;
\item asymptotically normal with $\sqrt{n}(\boldsymbol\beta^{(1)}_{n\mathcal{A}}-\boldsymbol\beta^{*}_{\mathcal{A}})\conind \mbox{N}({0},\sigma^{2}\mathbf{C}_{\mathcal{A}}^{-1})$.
\end{enumerate}
\end{thm}

The proofs are deferred to Section 8.
\par
\vspace{10pt}

\setcounter{chapter}{5}
\setcounter{equation}{0} %-1
\noindent {\bf 5. Experiments}

\noindent {\bf 5.1. Simulation}

In this section, we compare the proposed estimators to the posterior means obtained under the normal, Laplace, and horseshoe priors, to the Bayesian model averaged ($\small{\mbox{BMA}}$) estimator, as well as to the sparse estimates resulting from $\small{\mbox{LASSO}}$ (Tibshirani (1996)) and $\small{\mbox{SCAD}}$ (Fan and Li (2001)).  $\small{\mbox{GDP(PM)}}$ and $\small{\mbox{GDP(MAP)}}$ denote the posterior mean and the $\small{\mbox{MAP}}$ estimates, respectively, under the generalized double Pareto prior. Hyper-parameter values are provided in footnotes of Tables \ref{tab1} and \ref{tab2} when fixed in advance and are otherwise treated as random with the priors specified in Section 3. When not fixed, we first obtain the posterior means of the hyper-parameters from an initial Bayesian analysis, then use them in the calculation of the $\small{\mbox{MAP}}$ estimates.

We generated $n=\{50,400\}$ observations from  $y_{i}=\mathbf{x}'_{i}\boldsymbol\beta^*+\epsilon_{i}$, where the $x_{ij}$ were standard normals with $\mbox{Cov}({x}_{j},{x}_{j'})=0.5^{|j-j'|}$, $\epsilon_{i}\sim \mbox{N}(0,\sigma^{2})$, and $\sigma=3$. We used the following $\boldsymbol\beta^*$ configurations:

\noindent\emph{Model 1}: $5$ randomly chosen components of $\boldsymbol\beta^*$ set to $1$ and the rest to $0$.\\
\emph{Model 2}: $5$ randomly chosen components of $\boldsymbol\beta^*$ set to $3$ and the rest to $0$.\\
\emph{Model 3}: $10$ randomly chosen components of $\boldsymbol\beta^*$ set to $1$ and the rest to $0$.\\
\emph{Model 4}: $10$ randomly chosen components of $\boldsymbol\beta^*$ set to $3$ and the rest to $0$.\\
\emph{Model 5}: $\boldsymbol\beta^*=(0.85,\ldots,0.85)'$.

In our experiments $\mathbf{y}$ and the columns of $\mathbf{X}$ were centered and the columns of $\mathbf{X}$ scaled to have unit length.  For the calculation of competing estimators we used \texttt{lars} (Hastie and Efron (2011)), \texttt{SIS} (Fan et al. (2010)), \texttt{monomvn} (Gramacy (2010)) and \texttt{BAS} (Clyde and Littman (2005), Clyde, Ghosh, and Littman (2010)) packages in \texttt{R}. We mainly followed the default settings provided by the packages. Under the normal prior, the so-called ``ridge'' parameter was given an inverse gamma prior with shape and scale parameters $10^{-3}$. Under the Laplace prior, as a default choice, a gamma prior was placed on the ``$\small{\mbox{LASSO}}$ parameter'' $\lambda^2$, as given in (6) of Park and Casella (2008), with shape and rate parameters $2$ and $0.1$, respectively. Under the horseshoe prior, the \texttt{monomvn} package uses the hierarchy given in Section 1.1 of Carvalho, Polson, and Scott (2010). For $\small{\mbox{BMA}}$, we used the default settings of the \texttt{BAS} package that employs the Zellner-Siow prior given in Section 3$.$1 of Liang et al. (2008). The tuning for $\small{\mbox{LASSO}}$ and $\small{\mbox{SCAD}}$ were carried out by the criteria given in Yuan and Lin (2005) and  Wang, Li, and Tsai (2007), respectively, avoiding cross-validation.

$100$ data sets were generated for each case. In Table \ref{tab1}, we report the median model error. Model error was calculated as $(\boldsymbol\beta^*-\hat{\boldsymbol\beta})'\mathbf{C}(\boldsymbol\beta^*-\hat{\boldsymbol\beta})$, where $\mathbf{C}$ is the variance-covariance matrix that generated $\mathbf{X}$ and $\hat{\boldsymbol\beta}$ denotes the estimator in use. The values in the subscripts give the bootstrap standard error of the median model error values obtained. The bootstrap standard error was calculated by generating 500 bootstrap samples from 100 model error values, finding the median model error for each case, and then calculating the standard error for it. Under each model, the best three performances are boldfaced in the tables. 

$\small{\mbox{GDP(PM)}}$ estimates showed a similar performance to that of horseshoe under sparse setups. $\small{\mbox{GDP(PM)}}$ (with $\alpha$ and $\eta$ unknown) also showed great flexibility in adapting to dense models with small signals. $\small{\mbox{GDP(MAP)}}$ estimates performed similarly to $\small{\mbox{SCAD}}$ and much better than $\small{\mbox{LASSO}}$, particularly so with increasing sparsity, signal and/or sample size. The $\small{\mbox{GDP(PM)}}$ and $\small{\mbox{GDP(MAP)}}$ calculations are straightforward and computationally inexpensive due to the normal (and Laplace) scale mixture representation used. Being able to use a simple Gibbs sampler (especially when $\alpha=\eta=1$) makes the procedure attractive for the average user. 

Letting $\alpha=\eta=1$ may be somewhat restrictive if the underlying model is very dense or very sparse, but in the cases we considered, it performed comparably to others and we believe that it constitutes a good default prior similar to standard Cauchy with the added advantage of thresholding ability. Although we do not take up $p\gg n$ cases in this paper, in such situations much larger values of $\alpha$ would need to be chosen to adjust for multiplicity.

\begin{table}[!h]
\caption{Model error comparisons.}
\small
\label{tab1}
\begin{center}
\begin{threeparttable}[b]
\begin{tabular}{lrrrrr}
& \multicolumn{5}{c}{$n=50$} \\
\hline
\bf Method  & \bf Model 1 & \bf Model 2 & \bf Model 3 & \bf Model 4 & \bf Model 5\\
\hline
Normal				& $2.299_{0.085}$       & $4.879_{0.263}$    & $\bf2.585_{0.134}$ & $4.972_{0.385}$     & $\bf2.886_{0.150}$	\\
Laplace				& $2.634_{0.137}$       & $3.662_{0.233}$    & $\bf2.837_{0.126}$ & $4.326_{0.211}$     & $\bf3.458_{0.120}$  \\
Horseshoe			& $\bf2.264_{0.086}$    & $2.316_{0.167}$    & $3.205_{0.140}$    & $3.929_{0.218}$     & $4.409_{0.130}$	\\
$\small{\mbox{BMA}}$		& $2.451_{0.123}$       & $\bf1.647_{0.126}$ & $4.043_{0.233}$    & $\bf3.062_{0.194}$  & $6.015_{0.301}$	\\
$\small{\mbox{GDP(PM)}}^{1}$   	& $2.306_{0.114}$       & $2.405_{0.192}$    & $3.193_{0.215}$    & $4.123_{0.304}$     & $4.283_{0.142}$	\\
$\small{\mbox{GDP(PM)}}^{2}$ 	& $2.303_{0.095}$       & $2.309_{0.195}$    & $3.124_{0.153}$    & $3.910_{0.237}$     & $4.451_{0.109}$  \\
$\small{\mbox{GDP(PM)}}$    	& $\bf2.271_{0.085}$    & $2.606_{0.167}$    & $\bf3.047_{0.147}$ & $4.348_{0.171}$     & $\bf3.640_{0.134}$	\\
$\small{\mbox{GDP(MAP)}}^{1}$   & $3.414_{0.148}$       & $\bf1.619_{0.150}$ & $5.605_{0.298}$    & $\bf2.970_{0.168}$  & $8.769_{0.403}$	\\
$\small{\mbox{GDP(MAP)}}^{2}$   & $4.250_{0.354}$       & $\bf1.618_{0.153}$ & $6.331_{0.300}$    & $\bf3.040_{0.163}$  & $9.308_{0.377}$	\\
$\small{\mbox{GDP(MAP)}}$	& $4.876_{0.355}$       & $2.091_{0.182}$    & $4.299_{0.222}$    & $3.740_{0.284}$     & $5.724_{0.177}$	\\
$\small{\mbox{LASSO}}$	        & $\bf2.183_{0.124}$    & $2.618_{0.152}$    & $3.258_{0.194}$    & $3.531_{0.172}$     & $5.646_{0.229}$	\\
$\small{\mbox{SCAD}}$	        & $3.732_{0.214}$       & $2.132_{0.229}$    & $5.249_{0.239}$    & $3.179_{0.193}$     & $8.505_{0.387}$	\\
\hline
& \multicolumn{5}{c}{$n=400$} \\
\hline
Normal				& $0.395_{0.014}$    & $0.455_{0.019}$    & $0.426_{0.016}$    & $0.455_{0.024}$     & $\bf0.412_{0.013}$	\\
Laplace				& $0.315_{0.016}$    & $0.374_{0.014}$    & $0.388_{0.016}$    & $0.422_{0.015}$     & $\bf0.457_{0.014}$  \\
Horseshoe			& $0.219_{0.016}$    & $0.205_{0.010}$    & $0.341_{0.014}$    & $0.346_{0.009}$     & $0.514_{0.023}$	\\
$\small{\mbox{BMA}}$		& $\bf0.151_{0.011}$ & $0.125_{0.005}$    & $\bf0.240_{0.016}$ & $0.211_{0.009}$     & $0.646_{0.037}$	\\
$\small{\mbox{GDP(PM)}}^{1}$   	& $0.233_{0.016}$    & $0.206_{0.009}$    & $0.326_{0.015}$    & $0.284_{0.014}$     & $0.625_{0.031}$	\\
$\small{\mbox{GDP(PM)}}^{2}$ 	& $0.228_{0.017}$    & $0.215_{0.009}$    & $0.332_{0.013}$    & $0.303_{0.010}$     & $0.579_{0.027}$  \\
$\small{\mbox{GDP(PM)}}$    	& $0.248_{0.017}$    & $0.182_{0.007}$    & $0.377_{0.016}$    & $0.362_{0.012}$     & $\bf0.466_{0.016}$	\\
$\small{\mbox{GDP(MAP)}}^{1}$   & $\bf0.154_{0.014}$ & $\bf0.111_{0.011}$ & $0.286_{0.016}$    & $\bf0.210_{0.011}$  & $0.739_{0.043}$	\\
$\small{\mbox{GDP(MAP)}}^{2}$   & $0.161_{0.013}$    & $\bf0.111_{0.010}$ & $\bf0.284_{0.016}$ & $\bf0.210_{0.009}$  & $0.652_{0.035}$	\\
$\small{\mbox{GDP(MAP)}}$	& $0.185_{0.017}$    & $0.119_{0.010}$    & $0.326_{0.016}$    & $0.336_{0.010}$     & $0.478_{0.020}$	\\
$\small{\mbox{LASSO}}$	        & $0.251_{0.014}$    & $0.276_{0.014}$    & $0.339_{0.020}$    & $0.348_{0.011}$     & $0.485_{0.021}$	\\
$\small{\mbox{SCAD}}$	        & $\bf0.121_{0.010}$ & $\bf0.118_{0.008}$ & $\bf0.233_{0.011}$ & $\bf0.206_{0.017}$  & $0.469_{0.019}$	\\
\hline
\end{tabular}
\begin{tablenotes}
\footnotesize
\item $^{1}$$\alpha=1$, $\eta=1$; $^{2}$$\eta=1$
\end{tablenotes}
\end{threeparttable}
\end{center}
\end{table}
\par
\vspace{10pt}

\noindent {\bf 5.2. Inferences on Hyper-parameters}

Here we take a closer look at the inferences on the hyper-parameters obtained from an individual data set for Models 2 and 5 from Section 5.1. This gives us some insight into how $\alpha$ and $\eta$ are inferred with changing sample size and sparsity structure. Note that $\small{\mbox{GDP(PM)}}^{2}$ is more restrictive than $\small{\mbox{GDP(PM)}}$ as $\eta$ is fixed, treating only $\alpha$ as unknown. Figure \ref{fig5} gives the marginal posteriors of $\alpha$ and $\eta$ in cases of $\small{\mbox{GDP(PM)}}^{2}$ and $\small{\mbox{GDP(PM)}}$ as described in Section 5.1, while Table \ref{tab2} reports the posterior means for $\alpha$ and $\eta$, as well as model error (ME) performance (as calculated in Section 5.1) on the particular data set used. We clearly observe the adaptive nature and higher flexibility of $\small{\mbox{GDP(PM)}}$ moving from a sparse to a dense model with a big increase, particularly in $\eta$, flattening the prior on $\boldsymbol\beta$. There is not quite as much wiggle room in the case of $\small{\mbox{GDP(PM)}}^{2}$. All it can do is to drive $\alpha$ smaller to allow heavier tails to accommodate a dense structure. As observed in Table \ref{tab1}, however, $\small{\mbox{GDP(PM)}}^{2}$ performs comparably in sparse cases.

\begin{table}[!h]
\caption{Posterior means of the hyper-parameters and the resulting model error.}
\small
\label{tab2}
\begin{center}
\begin{threeparttable}[b]
\begin{tabular}{llrrrr}

& & \multicolumn{2}{c}{$\bf n=50$} & \multicolumn{2}{c}{$\bf n=400$}\\
\hline
& &$\small{\mbox{GDP(PM)}}$  &$\small{\mbox{GDP(PM)}}^{2}$ & $\small{\mbox{GDP(PM)}}$ & $\small{\mbox{GDP(PM)}}^{2}$\\
\hline
{\bf Model 2} & $\alpha$ & $2.464$ & $1.165$ & $0.688$ & $0.870$\\
& $\eta$     & $4.181$ &      --          & $0.614$ & \\
& ME          & $2.443$ & $2.219$   & $0.149$ & $0.181$\\
\hline
{\bf Model 5} & $\alpha$ & $5.262$ & $1.200$ & $9.400$ & $0.560$ \\
& $\eta$    & $9.476$ & --		    & $51.735$ & -- \\
& ME          & $6.290$ & $7.019$ & $0.518$ & $0.614$ \\
\hline
\end{tabular}
\begin{tablenotes}
\footnotesize
\item $^{2}$$\eta=1$
\end{tablenotes}
\end{threeparttable}
\end{center}
\end{table}

\begin{figure}[!h]
\centering \subfigure[]{
\begin{minipage}{.60\linewidth}
 \centering\includegraphics[width=.50\textwidth]{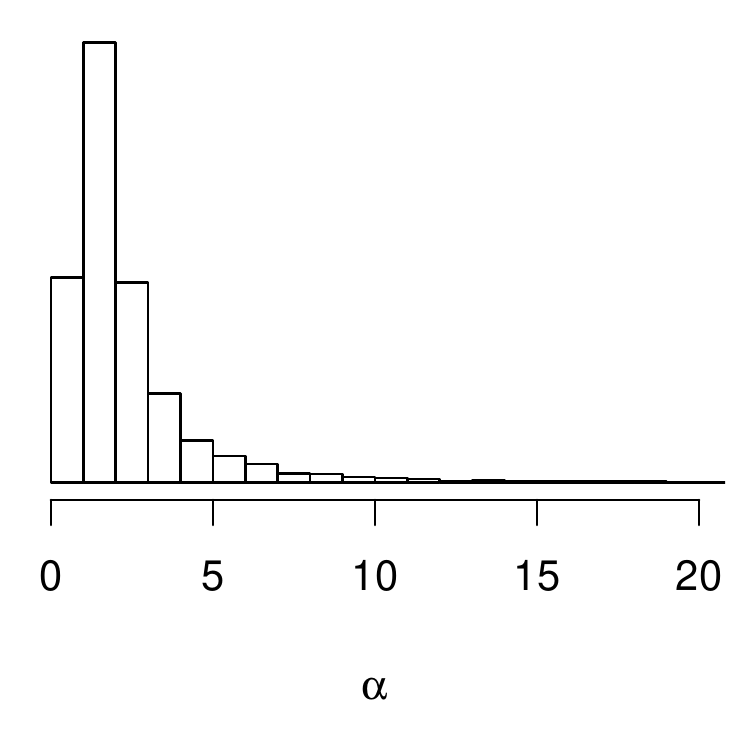}\includegraphics[width=.50\textwidth]{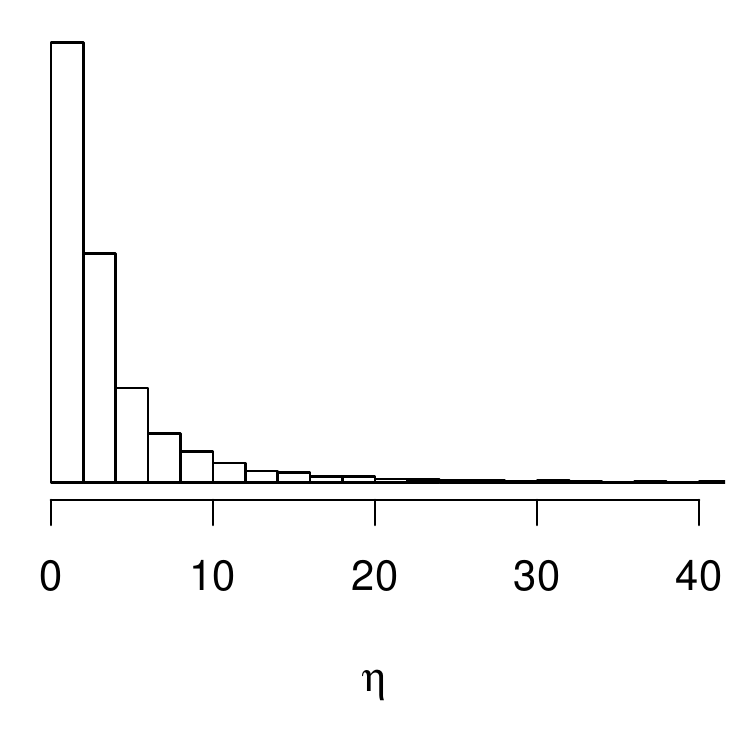}
\end{minipage}}
\centering \subfigure[]{
\begin{minipage}{.30\linewidth}
 \centering\includegraphics[width=1\textwidth]{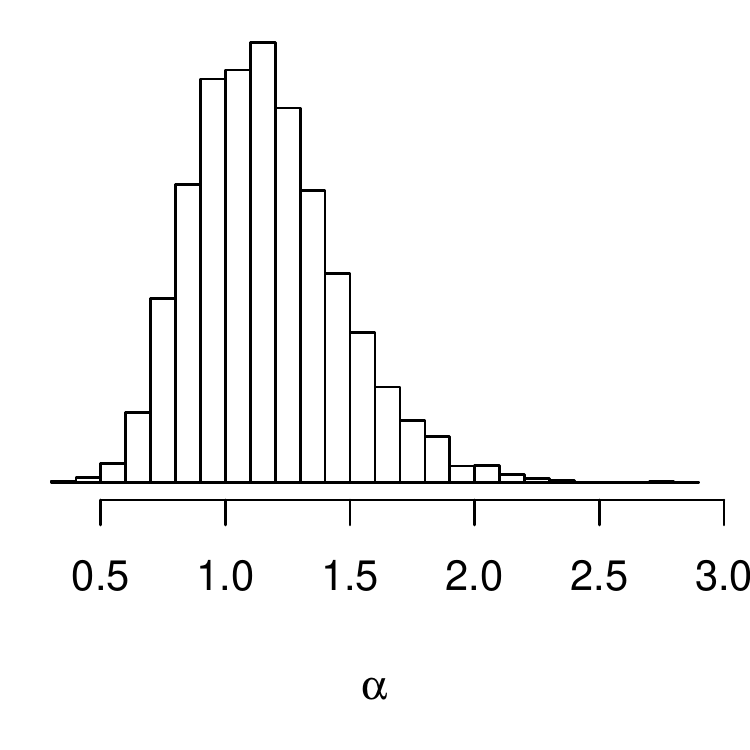}
\end{minipage}}
\centering \subfigure[]{
\begin{minipage}{.60\linewidth}
 \centering\includegraphics[width=.50\textwidth]{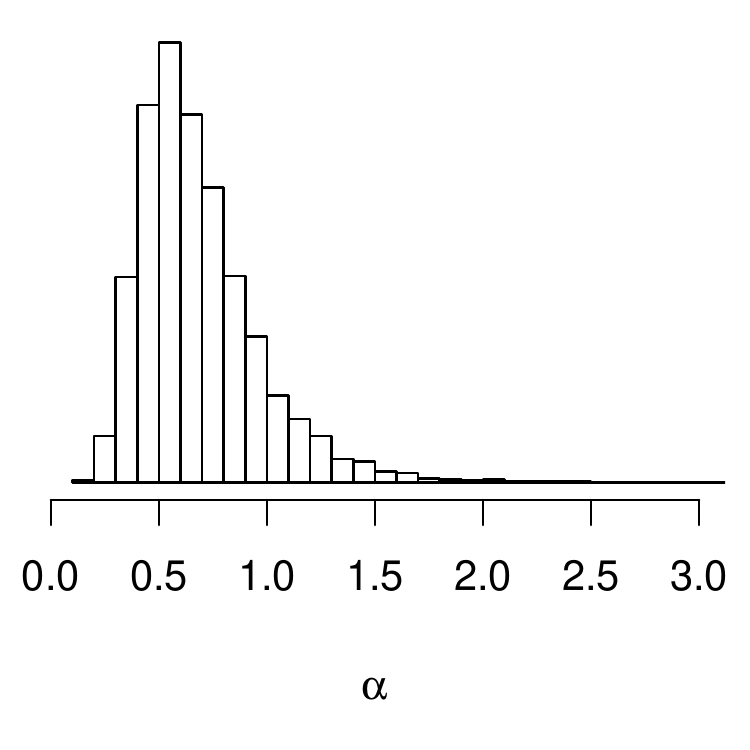}\includegraphics[width=.50\textwidth]{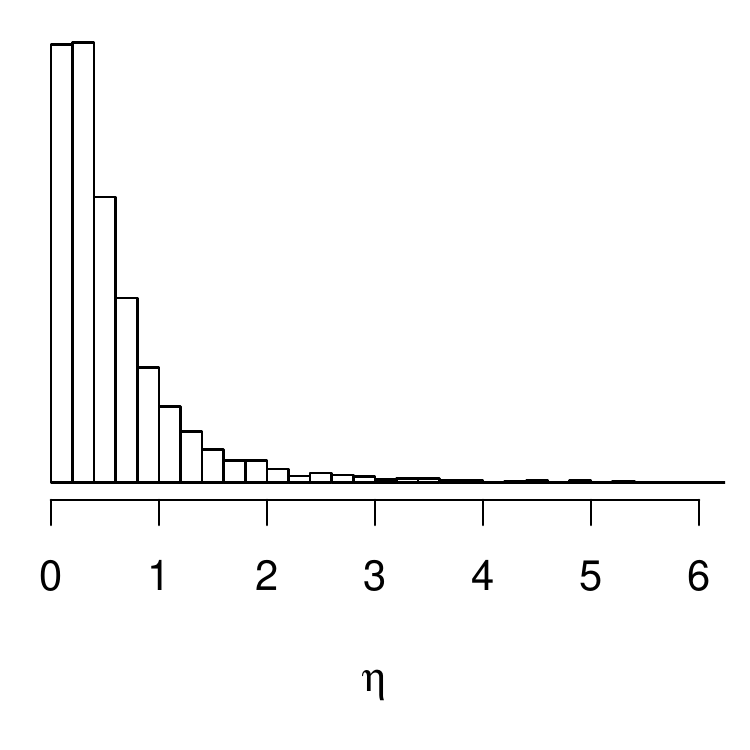}
\end{minipage}}
\centering \subfigure[]{
\begin{minipage}{.30\linewidth}
 \centering\includegraphics[width=1\textwidth]{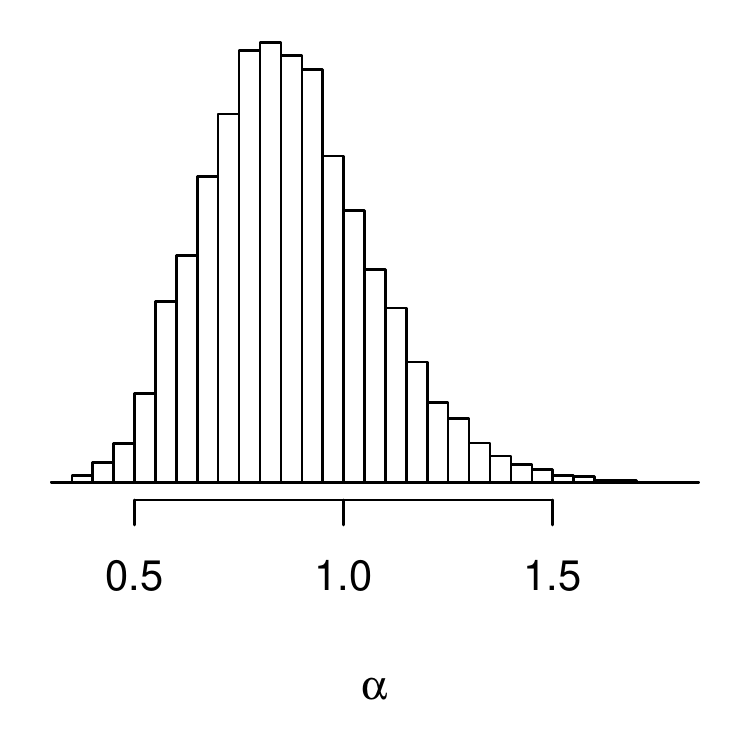}
\end{minipage}}
\centering \subfigure[]{
\begin{minipage}{.60\linewidth}
 \centering\includegraphics[width=.50\textwidth]{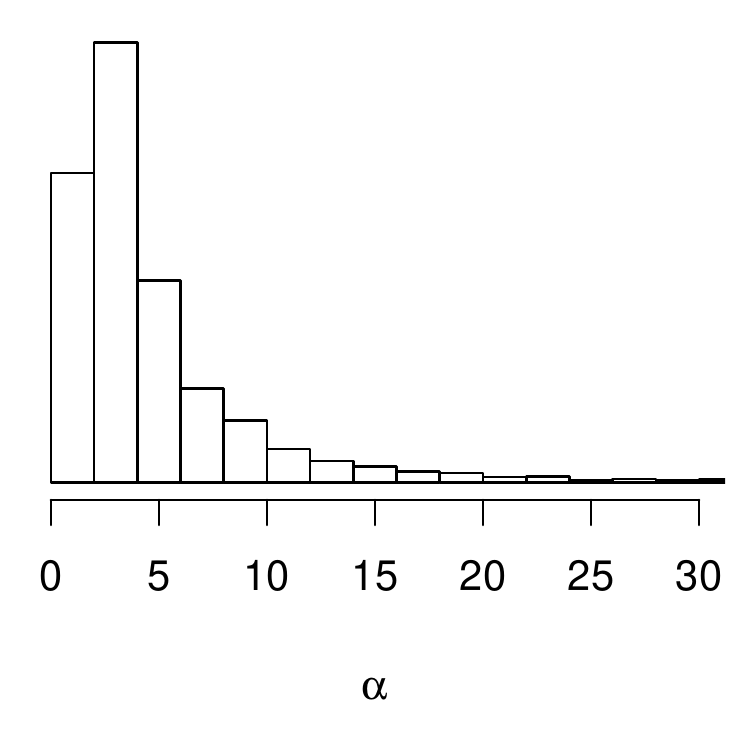}\includegraphics[width=.50\textwidth]{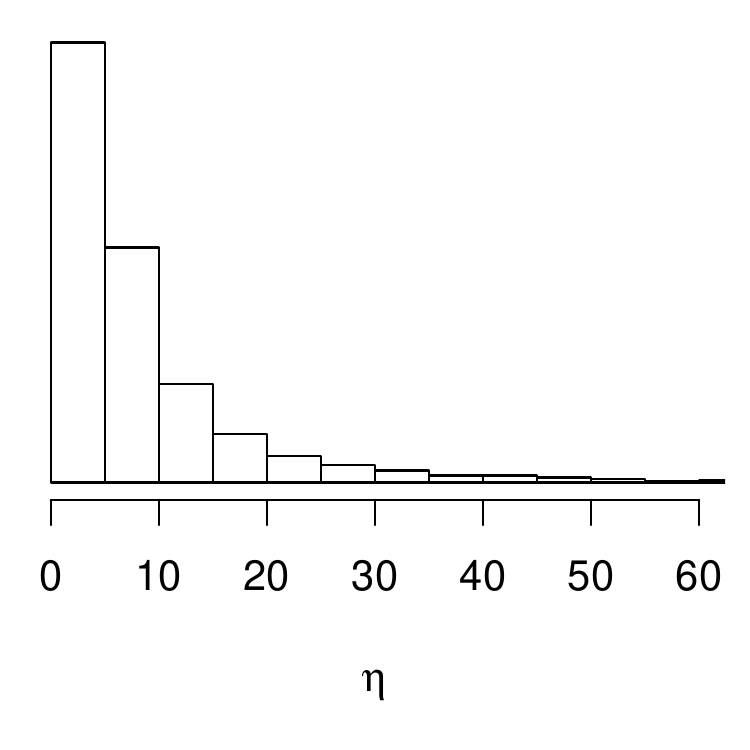}
\end{minipage}}
\centering \subfigure[]{
\begin{minipage}{.30\linewidth}
 \centering\includegraphics[width=1\textwidth]{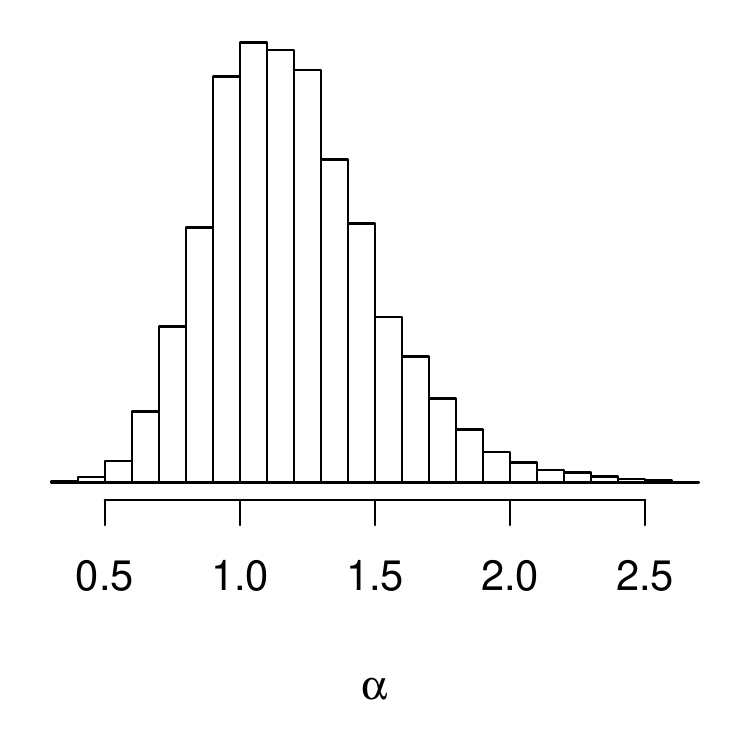}
\end{minipage}}
\centering \subfigure[]{
\begin{minipage}{.60\linewidth}
 \centering\includegraphics[width=.50\textwidth]{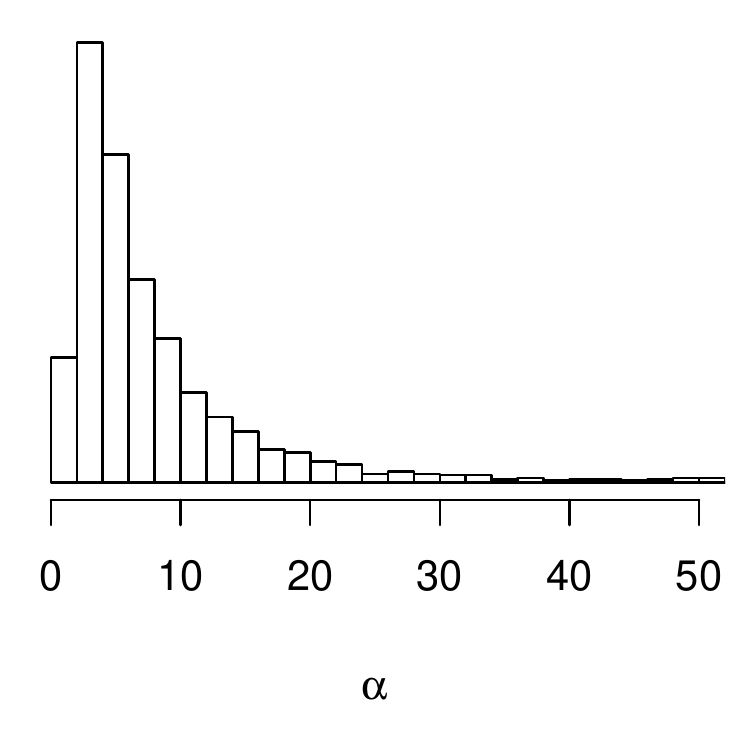}\includegraphics[width=.50\textwidth]{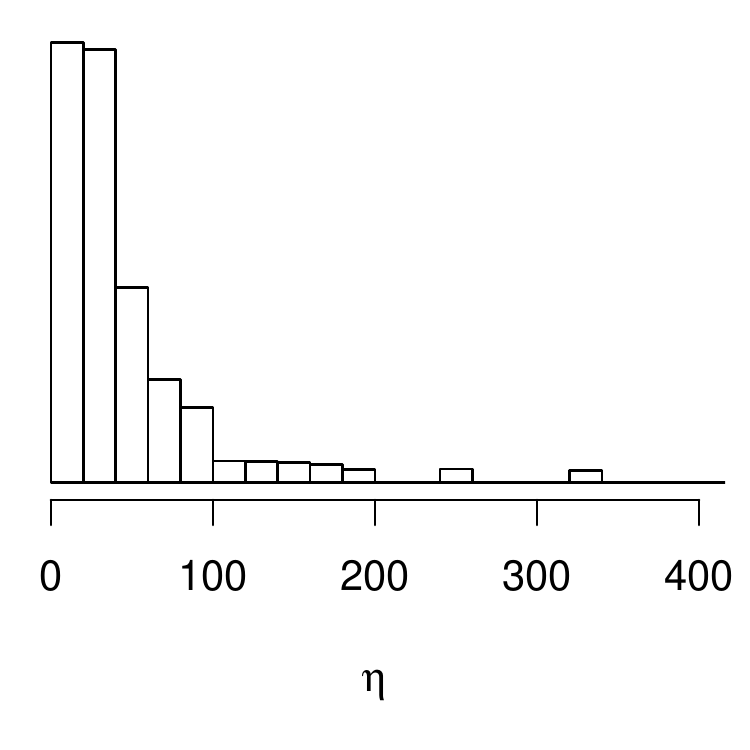}
\end{minipage}}
\centering \subfigure[]{
\begin{minipage}{.30\linewidth}
 \centering\includegraphics[width=1\textwidth]{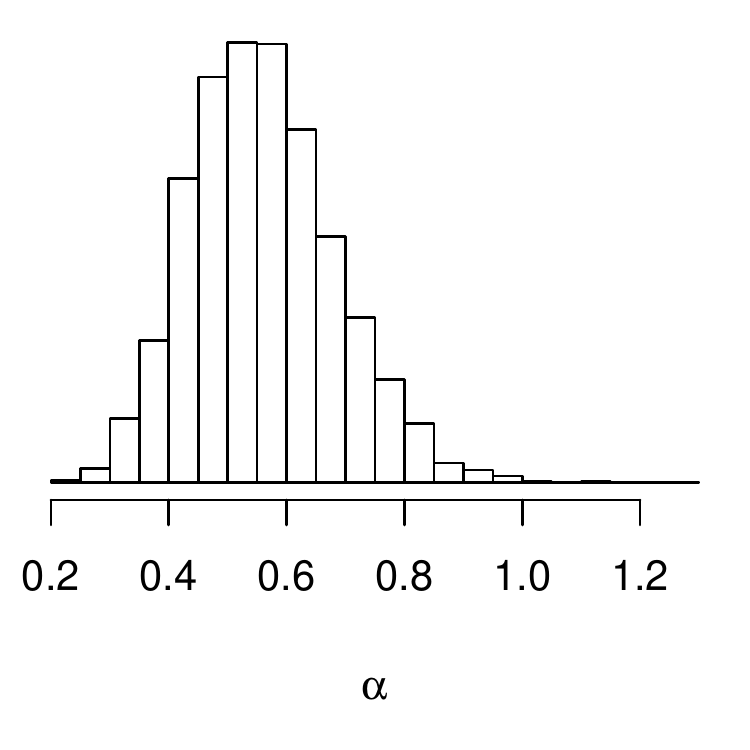}
\end{minipage}}
\caption{Inferences for (a) $\small{\mbox{GDP(PM)}}$ for $n=50$ under Model 2, (b) $\small{\mbox{GDP(PM)}}^{2}$ for $n=50$ under Model 2, (c) $\small{\mbox{GDP(PM)}}$ for $n=400$ under Model 2, (b) $\small{\mbox{GDP(PM)}}^{2}$ for $n=400$ under Model 2, (e) $\small{\mbox{GDP(PM)}}$ for $n=50$ under Model 5, (f) $\small{\mbox{GDP(PM)}}^{2}$ for $n=50$ under Model 5, (g) $\small{\mbox{GDP(PM)}}$ for $n=400$ under Model 5, (h) $\small{\mbox{GDP(PM)}}^{2}$ for $n=400$ under Model 2 .\label{fig5}}
\end{figure}
\par
\vspace{10pt}

\setcounter{chapter}{6}
\setcounter{equation}{0} %-1
\noindent {\bf 6.  Data Example}

We consider the ozone data analyzed by Breiman and Friedman (1985) and by Casella and Moreno (2006). The original data set contains 13 variables and 366 observations.
The modeled response is the daily maximum one-hour averaged ozone reading in Los Angeles over 330 days in 1976. There are $p=12$ predictors considered and deleting incomplete observations leaves $n=203$ observations.
For validation, the data were split into a training set containing 180 observations and a test set containing
23 observations. We considered models including main effects, quadratic, and two-way interaction terms resulting in $2^{90}$
possible subsets. The complex correlation structure of the data is illustrated in Figure \ref{fig6}.

\begin{figure}[!h]
 \centering\includegraphics[width=.6\textwidth]{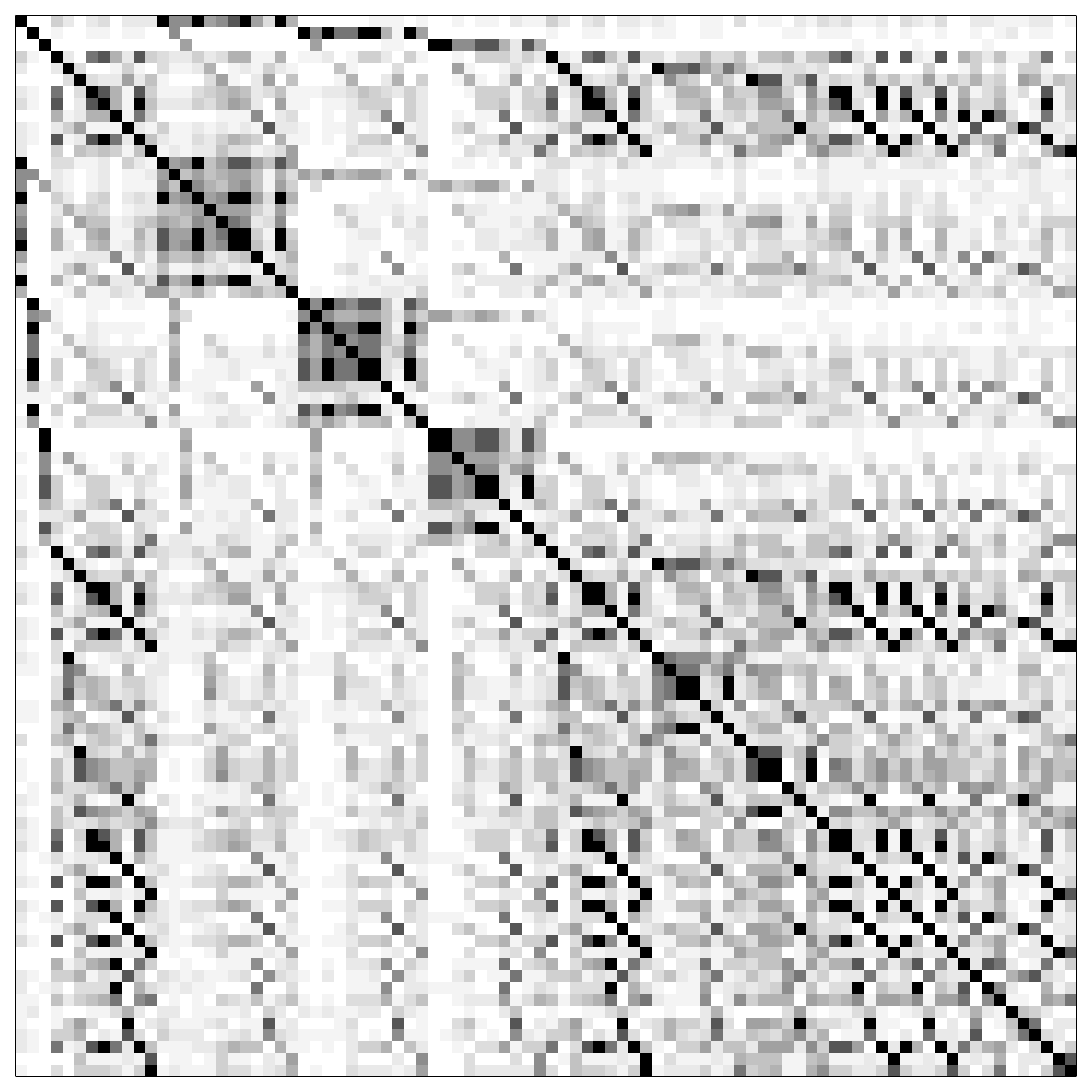}
\caption{The correlation structure of the Ozone data.  \label{fig6}} 
\end{figure}

Figure \ref{fig7} summarizes the performance of the proposed estimators and their competitors. Median values for $\mbox{R}^{2}_{\mbox{test}}$ and the $\pm2$ standard error intervals were obtained by running the methods on $100$ different random training-test splits. Standard errors were computed via bootstrapping the medians $500$ times.

\begin{figure}[!h]
 \centering\includegraphics[width=1\textwidth]{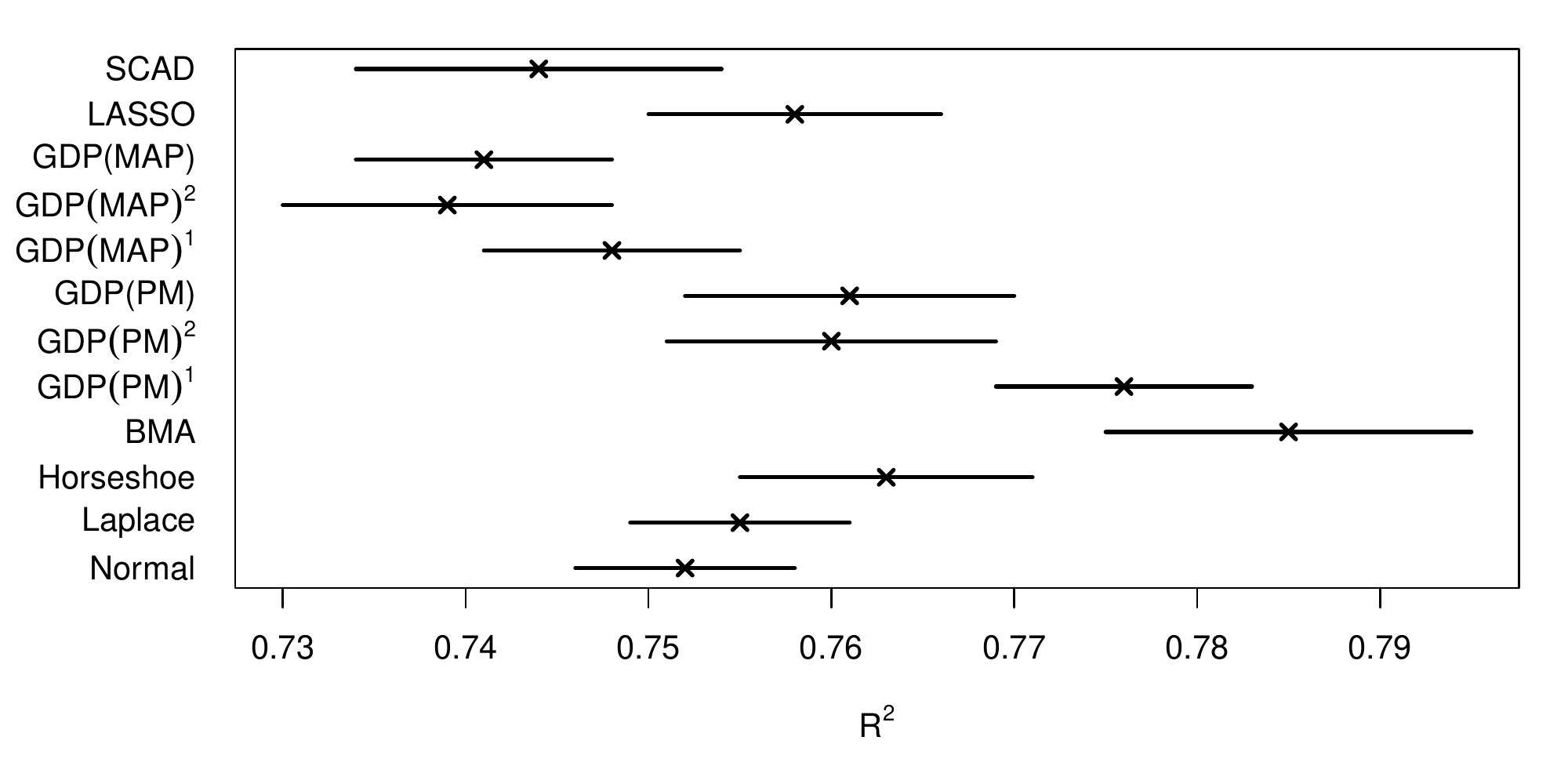}
\caption{Out-of-sample performance comparisons for Ozone data. ($\times$) denotes the median value for $\mbox{R}^{2}_{\mbox{test}}$ while the lines represent the $\pm2$ standard error regions. $^{1}$$\alpha=1$, $\eta=1$; $^{2}$$\eta=1$.  \label{fig7}} 
\end{figure}

%\begin{table}[!h]
%\centering \caption{Out-of-sample performance comparisons for Ozone data.}\vspace{10pt}
%\begin{threeparttable}[b]
% \begin{tabular}{lrr}
% \hline\noalign{\smallskip}
% \bf Method & $\mathbf{R^{2}_{test}}$ & \bf $|\mathcal{M}|$ \\
%\hline
%Normal				& $0.752_{0.006}$    & $90$    	\\
%Laplace				& $0.755_{0.006}$    & $90$      \\
%Horseshoe			& $0.763_{0.008}$    & $90$    	\\
%$\small{\mbox{BMA}}$		& $0.785_{0.010}$    & $90$    	\\
%$\small{\mbox{GDP(PM)}}^{1}$   	& $0.776_{0.007}$    & $90$    	\\
%$\small{\mbox{GDP(PM)}}^{2}$ 	& $0.760_{0.009}$    & $90$     \\
%$\small{\mbox{GDP(PM)}}$    	& $0.761_{0.009}$    & $90$    	\\
%$\small{\mbox{GDP(MAP)}}^{1}$   & $0.748_{0.007}$    & $4$ 	\\
%$\small{\mbox{GDP(MAP)}}^{2}$   & $0.739_{0.009}$    & $4$ 	\\
%$\small{\mbox{GDP(MAP)}}$	& $0.741_{0.007}$    & $4$    \\
%$\small{\mbox{LASSO}}$	        & $0.758_{0.008}$    & $14$    \\
%$\small{\mbox{SCAD}}$	        & $0.744_{0.010}$    & $9$ 	\\
%\hline
%\end{tabular}
%\begin{tablenotes}
%\footnotesize
%\item $^{1}$$\alpha=1$, $\eta=1$; $^{2}$$\eta=1$
%\end{tablenotes}
%\end{threeparttable}
%\label{tab3}
%\end{table}

The median number of predictors retained in the model by all three $\small{\mbox{GDP(MAP)}}$ estimates was only $4$ while it was $14$ and $9$ for $\small{\mbox{LASSO}}$ and $\small{\mbox{SCAD}}$. Hence $\small{\mbox{GDP(MAP)}}$ promoted much sparser models. In terms of prediction, $\small{\mbox{GDP(PM)}}^1$ yielded the second best results after $\small{\mbox{BMA}}$, with $\small{\mbox{GDP(PM)}}^2$, $\small{\mbox{GDP(PM)}}$, and the horseshoe estimator all having somewhat worse performance.  These shrinkage priors are designed to mimic model averaging behavior, so we expected to obtain results that were competitive with, but not better than, $\small{\mbox{BMA}}$.  The improved performance for $\small{\text{GDP(PM)}}^1$ may be attributed to the use of default hyper-parameter values that were fixed in advance at values thought to produce good performance in sparse settings.  Treating the hyper-parameters as unknown is appealing from the standpoint of flexibility, but in practice the data may not inform sufficiently about their values to outperform a good default choice. $\small{\mbox{GDP(MAP)}}^{1}$ and $\small{\mbox{SCAD}}$ both performed within the standard error range of $\small{\mbox{LASSO}}$, while retaining a smaller number of variables in the model.  As it is important to account for model uncertainty in prediction, the posterior mean estimator under the $\small{\mbox{GDP}}$ prior is appealing in mimicking $\small{\mbox{BMA}}$. In addition, obtaining a simple model containing a relatively small number of predictors is often important, since such models are more likely to be used in fields in which predictive black boxes are not acceptable and practitioners desire interpretable predictive models.
\par
\vspace{10pt}

\newpage

\setcounter{chapter}{7}
\setcounter{equation}{0} %-1
\noindent {\bf 7. Discussion}

We have proposed a hierarchical prior obtained through a particular scale mixture of normals where the resulting marginal prior has a folded generalized Pareto density thresholded at zero. This prior combines the best of both worlds in that fully Bayes inferences are feasible through its hierarchical representation, providing a measure of uncertainty in estimation, while the resulting marginal prior on the regression coefficients induces a penalty function that allows for the analysis of frequentist properties under maximum a posteriori estimation. The resulting posterior mean estimator can be argued to be mimicking a Bayesian model averaging behavior through mixing over higher level hyper-parameters. Although Bayesian model averaging is appealing, it can be argued that allowing parameters to be arbitrarily close to zero instead of exactly equal to zero may be more natural in some problems. Hence we have a procedure that not only bridges two paradigms -- Bayesian shrinkage estimation and regularization -- but also yields three useful tools: a sparse estimator with good frequentist properties through maximum a posteriori estimation, a posterior mean estimator that mimics a model averaging behavior, and a useful measure of uncertainty around the observed estimates. In addition, the proposed methods have substantial computational advantages in relying on simple block-updated Gibbs sampling, while $\small{\mbox{BMA}}$ requires sampling from a model space with $2^p$ models. Given the simple and fast computation and the excellent performance in small sample simulation studies, the generalized double Pareto should be useful as a shrinkage prior in a broad variety of Bayesian hierarchical models, while also suggesting close relationships with frequentist penalized likelihood approaches. The proposed prior can be used in generalized linear models, shrinkage of basis coefficients in nonparametric regression, and in such settings as factor analysis and nonparametric Bayes modeling. 
\par
\vspace{10pt}

\setcounter{chapter}{8}
\setcounter{equation}{0} %-1
\noindent {\bf 8. Technical Details}

\begin{proof}[Proof of Theorem 1]
The proof follows along similar lines as does the proof of Theorem 2 in Zou (2006). We first prove asymptotic normality. Let $\boldsymbol\beta=\boldsymbol\beta^{*}+\mathbf{u}/\sqrt{n}$ and 
\begin{equation*}
V_{n}(\mathbf{u})=\left\{\mathbf{y}-\sum_{j=1}^{p}\mathbf{x}_{j}\left(\beta_{j}^{*}+\frac{u_{j}}{\sqrt{n}}\right)\right\}^{2}+\alpha'_{n}\sum_{j=1}^{p}\log\left\{\left|\beta_{j}^{*}+\frac{u_{j}}{\sqrt{n}}\right|+\eta'_{n}\right\}.
\end{equation*} 
Let $\hat{\mathbf{u}}_{n}=\mbox{arg}\min V_{n}(\mathbf{u})$, suggesting $\hat{\mathbf{u}}_{n}=\sqrt{n}(\boldsymbol\beta^{(\infty)}_n-\boldsymbol\beta^{*})$. Now
\begin{equation*}
V_{n}(\mathbf{u})-V_{n}(\mathbf{0})=\mathbf{u}'\left(\frac{1}{n}\mathbf{X}'\mathbf{X}\right)\mathbf{u}-2\frac{\boldsymbol\epsilon'\mathbf{X}}{\sqrt{n}}\mathbf{u}+\alpha'_{n}\sum_{j=1}^{p}\log\frac{\left|\beta_{j}^{*}+u_{j}/\sqrt{n}\right|+\eta'_{n}}{\left|\beta_{j}^{*}\right|+\eta'_{n}},
\end{equation*}
and we know that $\mathbf{X}'\mathbf{X}/n\rightarrow \mathbf{C}$ and $\boldsymbol\epsilon'\mathbf{X}/\sqrt{n}\conind\mathbf{W}\eqd \mbox{N}(\mathbf{0},\sigma^{2}\mathbf{C})$. Consider the limiting behavior of the third term, noting that $\lim_{a\rightarrow \infty}(1+b/a)^{a}=\mbox{e}^b$. If $\beta_{j}^{*}\neq 0$, then $\alpha'_{n}\log\frac{|\beta_{j}^{*}+u_{j}/\sqrt{n}|+\eta'_{n}}{|\beta_{j}^{*}|+\eta'_{n}}\leq \alpha'_{n}\log\frac{|\beta_{j}^{*}|+|u_{j}/\sqrt{n}|+\eta'_{n}}{|\beta_{j}^{*}|+\eta'_{n}}=\alpha'_{n}\log\left(1+\frac{|u_{j}/\sqrt{n}|}{|\beta_{j}^{*}|+\eta'_{n}}\right)\rightarrow 0$. If $\beta_{j}^{*}=0$, then $\alpha'_{n}\log\frac{|u_{j}/\sqrt{n}|+\eta'_{n}}{\eta'_{n}}=\alpha'_{n}\log\left(1+\frac{|u_{j}/\sqrt{n}|}{\eta'_{n}}\right)$ which is $0$ if $u_{j}=0$, and diverges otherwise. By Slutsky's Theorem
\begin{equation*}
V_{n}(\mathbf{u})-V_{n}(\mathbf{0}) \conind \left\{\begin{array}{ll}\mathbf{u}'_{\mathcal{A}}\mathbf{C}_{\mathcal{A}}\mathbf{u}_{\mathcal{A}}-2\mathbf{u}'_{\mathcal{A}}\mathbf{W}_{\mathcal{A}} & \mbox{if}\ \ u_{j}=0 \ \ \forall j \notin \mathcal{A} \\
																 \infty & \mbox{otherwise}. 	\end{array}\right.
\end{equation*}
$V_{n}(\mathbf{u})-V_{n}(\mathbf{0})$ is convex and the unique minimum of the right hand side is $(\mathbf{C}_{\mathcal{A}}^{-1}\mathbf{W}_{\mathcal{A}},\mathbf{0})'$. By epiconvergence (Geyer (1994), Knight and Fu (2000)),
\begin{equation}
\hat{\mathbf{u}}_{n\mathcal{A}}\conind \mathbf{C}_{\mathcal{A}}^{-1}\mathbf{W}_{\mathcal{A}}, \ \ \ \
\hat{\mathbf{u}}_{n\mathcal{A}^{c}}\conind \mathbf{0}.
\label{asymp_norm}
\end{equation}
Since $\mathbf{W}_{\mathcal{A}}\eqd \mathrm{N}(\mathbf{0},\sigma^{2}\mathbf{C}_{\mathcal{A}})$, this proves asymptotic normality. 

Now $\forall j \in \mathcal{A}$, $\beta_{nj}^{(\infty)}\coninp \beta_{j}^{*}$; thus $\mathbb{P}(j\in \mathcal{A}_{n})\rightarrow 1$. Hence for consistency, it is sufficient to show that $\forall j' \notin \mathcal{A}$, $\mathbb{P}(j' \in \mathcal{A}_{n})\rightarrow 0$. Consider the event $j' \in \mathcal{A}_{n}$. By the KKT optimality conditions, $2\mathbf{x}_{j'}'(\mathbf{y}-\mathbf{X}\boldsymbol\beta_n^{(\infty)})=\frac{\alpha'_{n}}{\eta'_{n}+|\beta_{nj'}^{(\infty)}|}$. Noting that $\sqrt{n}\beta_{nj'}^{(\infty)}\coninp 0$ by (\ref{asymp_norm}), $\frac{\alpha'_{n}}{\sqrt{n}\eta'_{n}+\sqrt{n}|\beta_{nj'}^{(\infty)}|}\rightarrow \infty$,  while
\begin{equation*}
\frac{2\mathbf{x}_{j'}'(\mathbf{y}-\mathbf{X}\boldsymbol\beta_n^{(\infty)})}{\sqrt{n}}=2\left\{\frac{\mathbf{x}_{j'}'\mathbf{X}\sqrt{n}(\boldsymbol\beta^{*}-\boldsymbol\beta_n^{(\infty)})}{n}+\frac{\mathbf{x}_{j'}'\boldsymbol\epsilon}{\sqrt{n}}\right\}. 
\end{equation*}
By (\ref{asymp_norm}) and Slutsky's Theorem, we know that both terms in the brackets converge in distribution to some normal, so
\begin{equation*}
\mathbb{P}\left(j'\in \mathcal{A}_{n}\right)\leq \mathbb{P}\left\{2\mathbf{x}_{j'}'\left(\mathbf{y}-\mathbf{X}\boldsymbol\beta_n^{(\infty)}\right)=\frac{\alpha_{n}'}{|\beta^{(\infty)}_{nj'}|+\eta'_{n}} \right\} \rightarrow 0.
\end{equation*}
This concludes the proof.
\end{proof}

\begin{proof}[Proof of Theorem 2]
We modify the proof of Theorem 2 in Zou (2006). Here $\boldsymbol\beta^{(0)}_{n}$ denotes the least squares estimator. We first prove asymptotic normality. Let $\boldsymbol\beta=\boldsymbol\beta_n^{*}+\mathbf{u}/\sqrt{n}$ and 
\begin{equation*}
V_{n}(\mathbf{u})=\left\{\mathbf\mathbf{y}-\sum_{j=1}^{p}{x}_{j}\left(\beta_{nj}^{*}+\frac{u_{j}}{\sqrt{n}}\right)\right\}^{2}+\alpha^{\dag}_{n}\sum_{j=1}^{p}|\beta_{nj}^{*}+\frac{u_{j}}{\sqrt{n}}|\left(|\beta_{nj}^{(0)}|+\eta^{\dag}_{n}\right)^{-1}.
\end{equation*} 
Let $\hat{\mathbf{u}}_n=\mbox{arg}\min V_{n}(\mathbf{u})$, suggesting $\hat{\mathbf{u}}_{n}=\sqrt{n}(\boldsymbol\beta_n^{(1)}-\boldsymbol\beta_n^{*})$. Now
\begin{eqnarray*}
V_{n}(\mathbf{u})-V_{n}(\mathbf{0})&=&\mathbf{u}'\left(\frac{1}{n}\mathbf{X}'\mathbf{X}\right)\mathbf{u}-2\frac{\boldsymbol\epsilon'\mathbf{X}}{\sqrt{n}}\mathbf{u}\nonumber\\
& &+\frac{\alpha^{\dag}_{n}}{\sqrt{n}}\sum_{j=1}^{p}\left(|\beta_{nj}^{(0)}|+\eta^{\dag}_{n}\right)^{-1}\sqrt{n}\left(\left|\beta_{nj}^{*}+\frac{u_{j}}{\sqrt{n}}\right|-|\beta_{nj}^{*}|\right),
\end{eqnarray*}
and we know that $\mathbf{X}'\mathbf{X}/n\rightarrow \mathbf{C}$ and $\boldsymbol\epsilon'\mathbf{X}/\sqrt{n}\conind\mathbf{W}\eqd \mathrm{N}({0},\sigma^{2}\mathbf{C})$. Consider the limiting behavior of the third term. If $\beta_{nj}^{*}\neq 0$ then, by the Continuous Mapping Theorem, $\{|\beta_{nj}^{(0)}|+\eta^{\dag}_{n}\}^{-1}\coninp \{|\beta_{nj}^{*}|+\eta^{\dag}_{n}\}^{-1}$ and $\sqrt{n}(|\beta_{nj}^{*}+u_{j}/\sqrt{n}|-|\beta_{nj}^{*}|)\rightarrow u_{j}\mbox{sgn}(\beta_{nj}^{*})$. By Slutsky's Theorem, $(\alpha^{\dag}_{n}/\sqrt{n})\{|\beta_{nj}^{(0)}|+\eta^{\dag}_{n}\}^{-1}\sqrt{n}(|\beta_{nj}^{*}+u_{j}/\sqrt{n}|-|\beta_{nj}^{*}|)\coninp 0$. If $\beta_{nj}^{*}=0$, then $\sqrt{n}(|\beta_{nj}^{*}+u_{j}/\sqrt{n}|-|\beta_{nj}^{*}|)=|u_{j}|$ and $\alpha^{\dag}_{n}\{|\beta_{nj}^{(0)}|+\eta^{\dag}_{n}\}^{-1}/\sqrt{n}=\alpha^{\dag}_{n}/(\sqrt{n}|\beta_{nj}^{(0)}|+\sqrt{n}\eta^{\dag}_{n})$, where $\sqrt{n}\beta_{nj}^{(0)}=O_{p}(1)$. Again by Slutsky's Theorem,
\begin{equation*}
V_{n}(\mathbf{u})-V_{n}({0}) \conind \left\{\begin{array}{ll}\mathbf{u}'_{\mathcal{A}}\mathbf{C}_{\mathcal{A}}\mathbf{u}_{\mathcal{A}}-2\mathbf{u}'_{\mathcal{A}}\mathbf{W}_{\mathcal{A}} & \mbox{if}\ \ u_{j}=0 \ \ \mbox{for all} \ \ j \notin \mathcal{A} \\
																 \infty & \mbox{otherwise}. 	\end{array}\right.
\end{equation*}
$V_{n}(\mathbf{u})-V_{n}(\mathbf{0})$ is convex and the unique minimum of the right hand side is $(\mathbf{C}_{\mathcal{A}}^{-1}\mathbf{W}_{\mathcal{A}},\mathbf{0})'.$ By epiconvergence (Geyer (1994), Knight and Fu (2000)),
\begin{equation}
\hat{\mathbf{u}}_{n\mathcal{A}}\conind \mathbf{C}_{\mathcal{A}}^{-1}\mathbf{W}_{\mathcal{A}}, \ \ \ \ \hat{\mathbf{u}}_{n\mathcal{A}^{c}}\conind \mathbf{0}. 
\label{asymp_norm2}
\end{equation}
Since $\mathbf{W}_{\mathcal{A}}\eqd \mathrm{N}(\mathbf{0},\sigma^{2}\mathbf{C}_{\mathcal{A}})$, this proves the asymptotic normality. 

Now $\forall j \in \mathcal{A}$, $\beta_{nj}^{(1)}\coninp \beta_{nj}^{*}$; thus $\mathbb{P}(j\in \mathcal{A}_{n})\rightarrow 1$. We show that for all $j' \notin \mathcal{A}$, $\mathbb{P}(j' \in \mathcal{A}_{n})\rightarrow 0$. Consider the event $j' \in \mathcal{A}_{n}$. By the KKT optimality conditions, $2\mathbf{x}_{j'}'(\mathbf{y}-\mathbf{X}\boldsymbol\beta_{n}^{(1)})=\alpha_{n}^{\dag}(|\beta_{nj'}^{(0)}|+\eta_{n}^{\dag})^{-1}$. We know that $\alpha^{\dag}_{n}(|\beta_{nj'}^{(0)}|+\eta^{\dag}_{n})^{-1}/\sqrt{n}\coninp \infty$, while
\begin{equation*}
\frac{2\mathbf{x}_{j'}'(\mathbf{y}-\mathbf{X}\boldsymbol\beta_{n}^{(1)})}{\sqrt{n}}=2\left\{\frac{\mathbf{x}_{j'}'\mathbf{X}\sqrt{n}(\boldsymbol\beta_n^{*}-\boldsymbol\beta_{n}^{(1)})}{n}+\frac{\mathbf{x}_{j'}'\boldsymbol\epsilon}{\sqrt{n}}\right\}. 
\end{equation*}
By (\ref{asymp_norm2}) and Slutsky's Theorem, we know that both terms in the brackets converge in distribution to some normal, so
\begin{equation*}
\mathbb{P}\left(j'\in \mathcal{A}_{n}\right)\leq \mathbb{P}\left\{2\mathbf{x}_{j'}'\left(\mathbf{y}-\mathbf{X}\boldsymbol\beta_{n}^{(1)}\right)=\frac{\alpha^{\dag}_{n}}{|\beta_{nj'}^{(0)}|+\eta^{\dag}_{n}}\right\}\rightarrow 0,
\end{equation*}
which proves consistency.
\end{proof}

%%%%%%%%%%%%%%%%%%%%%%%%%%%%%%%%%%%%%%%%%%%%%%%%%%%%%%%%%%%%%%%%%%%%%%%%%%%%%%%%%%%%%%%%%%%%%%%%%%%%%%%%%%%%%%%%%%%%%%%%%%%%

\noindent {\bf Acknowledgments}

This work was supported by Award Number R01ES017436 from the National Institute of Environmental Health Sciences. The content is solely the responsibility of the authors and does not necessarily represent the official views of the National Institute of Environmental Health Sciences or the National Institutes of Health. Jaeyong Lee was supported by Basic Science Research Program through the National Research Foundation of Korea (NRF) funded by the Ministry of Education, Science and Technology (20110027353).
\par
\vspace{10pt}
%%%%%%%%%%%%%%%%%%%%%%%%%%%%%%%%%%%%%%%%%%%%%%%%%%%%%%%%%%%%%%%%%%%%%%%%%%%%%%%%%%%%%%%%%%%%%%%%%%%%%%%%%%%%%%%%%%%%%%%%%%%%

\noindent{\bf References}
\begin{description}
\item
Berger, J. (1980). A robust generalized Bayes estimator and confidence region
  for a multivariate normal mean.
  {\it The Annals of Statistics} {\bf 8}, 716--761.

\item
Berger, J. (1985). {\em Statistical Decision Theory and Bayesian Analysis}. Springer, New York.

\item
Breiman, L. (1996). Heuristics of instability and stabilization in model selection. {\it The Annals of Statistics} {\bf 24}, 2350--2383.

\item
Breiman, L. and Friedman, J. H. (1985), Estimating optimal transformations for multiple regression and correlation. {\it Journal of the American Statistical Association} {\bf 80}.

\item
Candes, E.~J., Wakin, M.~B., and Boyd, S.~P. (2008). Enhancing sparsity by reweighted $\ell_1$ minimization. {\it Journal of Fourier Analysis and Applications} {\bf 14}, 877--905.

\item
Carvalho, C., Polson, N., and Scott, J. (2009). Handling sparsity via the horseshoe. {\it JMLR: W\&CP} {\bf 5}.

\item
Carvalho, C., Polson, N., and Scott, J. (2010). The horseshoe estimator for sparse signals. {\it Biometrika} {\bf 97}, 465--480.

\item
Casella, G. and Moreno, E. (2006). Objective Bayesian variable selection. {\it Journal of the American Statistical Association} {\bf 101}.

\item
Cevher, V. (2009). Learning with compressible priors. {\it Advances in Neural Information Processing Systems} {\bf 22}.

\item
Clyde, M., Ghosh, J., and Littman, M.~L. (2010). Bayesian adaptive sampling for variable selection and model averaging. {\it Journal of Computational and Graphical Statistics} {\bf 20}, 80--101.

\item
Clyde, M. and Littman, M. (2005). {\it Bayesian model averaging using Bayesian adaptive sampling -- BAS package manual}.

\item 
Efron, B., Hastie, T., Johnstone, I., and Tibshirani, R. (2004). Least angle regression. {\em The Annals of Statistics\/} {\bf 32}, 407--499.

\item
Fan, J., Feng, Y., Samworth, R., and Wu, Y. (2010). {\it SIS package manual}.

\item
Fan, J. and Li, R. (2001). Variable selection via nonconcave penalized likelihood and its oracle properties. {\it Journal of the American Statistical Association} {\bf 96}, 1348--1360.

\item
Figueiredo, M. A.~T. (2003). Adaptive sparseness for supervised learning. {\it IEEE Transactions on Pattern Analysis and Machine Intelligence} {\bf 25}, 1150--1159.

\item
Fu, W. (1998). Penalized regressions: The bridge versus the lasso. {\it Journal of Computational and Graphical Statistics} {\bf 7}, 397--416.

\item
Garrigues, P.~J. (2009). Sparse coding models of natural images: Algorithms for efficient inference and learning of higher-order structure. PhD Thesis, University of California, Berkeley.

\item
Geyer, C.~J. (1994). On the asymptotics of constrained M-estimation. {\it The Annals of Statistics} {\bf 22}, 1993--2010.

\item
Gramacy, R.~B. (2010). {\it Estimation for multivariate normal and Student-t data with monotone missingness -- Monomvn package manual}.

\item
Griffin, J.~E. and Brown, P.~J. (2007). Bayesian adaptive lassos with non-convex penalization. Technical Report.

\item
Griffin, J.~E. and Brown, P.~J. (2010). Inference with normal-gamma prior distributions in regression problems. {\it Bayesian Analysis} {\bf 5}, 171--188.

\item
Hans, C. (2009). Bayesian lasso regression. {\it Biometrika} {\bf 96}, 835--845.

\item
Hastie, T. and Efron, B. (2011). {\it Lars package manual}.

\item
Knight, K. and Fu, W. (2000). Asymptotics for lasso-type estimators. {\it The Annals of Statistics} {\bf 28}, 1356--1378.

\item
Lee, A., Caron, F., Doucet, A., and Holmes, C. (2012). Bayesian sparsity-path-analysis of genetic association signal using generalized t priors. {\it Statistical Applications in Genetics
and Molecular Biology} {\bf 11}.

\item
Liang, F., Paulo, R., Molina, G., Clyde, M., and Berger, J. (2008). Mixtures of \textit{g} priors for Bayesian variable selection. {\it Journal of the American Statistical Association} {\bf 103}, 410--423.

\item
McDonald, J.~B. and Newey, W.~K. (1988). Partially adaptive estimation of regression models via the generalized t distribution. {\it Econometric Theory} {\bf 4}, 428--457.

\item
Park, T. and Casella, G. (2008). The {B}ayesian lasso. {\it Journal of the American Statistical Association} {\bf 103}, 681--686.

\item
Pickands, J. (1975). Statistical inference using extreme order statistics. {\it The Annals of Statistics} {\bf 3}, 119--131.

\item
Polson, N.~G. and Scott, J.~G. (2010). Shrink globally, act locally: Sparse Bayesian regularization and prediction. {\it Bayesian Statistics 9}. Oxford University Press.

\item
Ritter, C. and Tanner, M.~A. (1992). Facilitating the Gibbs sampler: The Gibbs stopper and the griddy-Gibbs sampler. {\it Journal of the American Statistical Association} {\bf 97}, 861--868.

\item
Strawderman, W.~E. (1971). Proper Bayes minimax estimators of the multivariate normal mean. {\it The Annals of Mathematical Statistics} {\bf 42}, 385--388.

\item
Tibshirani, R. (1996). Regression shrinkage and selection via the lasso. {\it Journal of the Royal Statistical Society } {\bf 58}, 267--288.

\item
Tipping, M.~E. (2001). Sparse Bayesian learning and the relevance vector machine. {\it Journal of Machine Learning Research} {\bf 1}.

\item
Wang, H., Li, R., and Tsai, C.~L. (2007). Tuning parameter selectors for the smoothly clipped absolute deviation method. {\it Biometrika} {\bf 94}, 553--568.

\item
West, M. (1987). On scale mixtures of normal distributions. {\it Biometrika} {\bf 74}, 646--648.

\item
Yuan, M. and Lin, Y. (2005). Efficient empirical Bayes variable selection and estimation in linear models. {\it Journal of the American Statistical Association} {\bf 100}, 1215--1225.

\item
Zhao, P. and Yu, B. (2006). On model selection consistency of lasso. {\it Journal of Machine Learning Research} {\bf 7}.

\item
Zou, H. (2006). The adaptive lasso and its oracle properties. {\it Journal of the American Statistical Association\/} {\bf 101}, 1418--1429.

\item
Zou, H. and Li, R. (2008). One-step sparse estimates in nonconcave penalized likelihood models. {\it The Annals of Statistics} {\bf 36}, 1509--1533.

\end{description}

%%%%%%%%%%%%%%%%%%%%%%%%%%%%%%%%%%%%%%%%%%%%%%%%%%%%%%%%%%%%%%%%%%%%%%%%%%%%%%%%%%%%%%%%%%%%%%%%%%%%%%%%%%%%%%%%%%%%%%%%%%%%

\vskip .65cm
\noindent
SAS Institute Inc., Durham, NC 27513, USA
\vskip 2pt
\noindent
E-mail: artin.armagan@sas.com
\vskip 2pt

\noindent
Department of Statistical Science, Duke University, Durham, NC 27708, USA
\vskip 2pt
\noindent
E-mail: dunson@stat.duke.edu
\vskip 2pt

\noindent
Department of Statistics, Seoul National University, Seoul, 151-747, Korea
\vskip 2pt
\noindent
E-mail: leejyc@gmail.com
\vskip .3cm
%%%%%%%%%%%%%%%%%%%%%%%%%%%%%%%%%%%%%%%%%%%%%%%%%%%%%%%%%%%%%%%%%%%%%%%%%%%%%%%%%%%%%%%%%%%%%%%%%%%%%%%%%%%%%%%%%%%%%%%%%%%%
%%%%%%%%%%%%%%%%%%%%%%%%%%%%%%%%%%%%%%%%%%%%%%%%%%%%%%%%%%%%%%%%%%%%%%%%%%%%%%%%%%%%%%%%%%%%%%%%%%%%%%%%%%%%%%%%%%%%%%%%%%%%

\end{document}